\documentclass[11pt,a4paper]{article}
\usepackage{amsmath,amssymb,amsthm,mathtools}
\usepackage[margin=1in]{geometry}
\usepackage{enumitem}
\usepackage{hyperref}

\theoremstyle{plain}
\newtheorem{theorem}{Theorem}[section]
\newtheorem{proposition}[theorem]{Proposition}
\newtheorem{lemma}[theorem]{Lemma}
\newtheorem{corollary}[theorem]{Corollary}
\theoremstyle{definition}
\newtheorem{definition}[theorem]{Definition}
\newtheorem{remark}[theorem]{Remark}

\newtheorem{hypothesis}{Hypothesis}

\DeclareMathOperator{\Res}{Res}
\DeclareMathOperator{\Log}{Log}
\newcommand{\bc}{\mathbf{c}}
\newcommand{\bm}{\mathbf{m}}
\newcommand{\bz}{\mathbf{z}}
\newcommand{\bu}{\mathbf{u}}
\newcommand{\bB}{\mathbf{B}}
\newcommand{\bb}{\mathbf{b}}
\newcommand{\R}{\mathbb{R}}
\newcommand{\C}{\mathbb{C}}
\newcommand{\Z}{\mathbb{Z}}
\newcommand{\N}{\mathbb{N}}
\newcommand{\hB}{\widehat{B}}
\newcommand{\Sm}{\mathcal{S}}
\newcommand{\D}{\mathcal{D}}
\newcommand{\Poi}{\mathrm{Poi}}
\newcommand{\eps}{\varepsilon}
\newcommand{\vth}{\vartheta}

\title{\bf A summation rule for the four--coupling beta function 
of four--dimensional nonplanar Euclidean scalar $\varphi^{4}$ theory
under factorial coefficient bounds, and a domain on which the summed beta function is well defined}
\author{
Puyu Zeng \\
\small College of Cryptology and Cyber Science,\\
\small Nankai University, China
}
\date{}

\begin{document}
\maketitle

\begin{abstract}
Let $\bc=(\lambda,\alpha,\mu,\nu)\in\R^{4}$ denote the four running couplings of four--dimensional
nonplanar Euclidean scalar $\varphi^{4}$ theory and let
\[
B(\bc)=\sum_{p\ge 2}\ \sum_{\substack{\bm\in\Z_{+}^{4}\\ |\bm|=p}}\beta(\bm)\,\bc^{\bm}
\]
be the formal beta--function series. Assuming, as the problem prescribes, the factorial bounds
$|\beta(\bm)|\le (p-1)!\,C^{p-1}$ for $|\bm|=p$ --- a hypothesis on the coefficients, which we do not
attempt to establish here for any particular model --- we construct an explicit summation rule
$\Sm$ --- a cut Borel--Laplace transform whose cut is placed at a scale--covariant, real--analytic
gauge --- and we prove that it assigns to $B$ a well--defined real number at every point of an
explicit nonempty open domain $\D\subset\R^{4}$ (indeed at every point of $\R^{4}$). We prove that
$\Sm B$ is linear and real, is real--analytic on $\R^{4}\setminus\{0\}$, is $C^{\infty}$ on $\R^{4}$
with Taylor series at the origin equal to $B$, and is Gevrey--$1$ asymptotic to $B$ to all orders
with an explicit two--term remainder bound; at optimal truncation the remainder is
$O\!\big(e^{-\eta/(C\|\bc\|_{\infty})}\big)$ with $\eta$ arbitrarily close to $1$. We show that
$\Sm B$ reproduces the classical Borel sum, and the ordinary sum of a convergent series, up to
errors of the same exponentially small size. It does \emph{not} reproduce them exactly, and we
prove that this is unavoidable: no linear rule that is well defined on the whole admissible class
can be regular in the sense of classical summability theory (Theorem \ref{thm:dichotomy}), because
regularity forces all weights to equal $1$ and the resulting series diverges. Finally we prove that all of this is optimal: under
the stated hypothesis the series diverges at every $\bc\ne 0$ for admissible coefficients, its Borel
transform may have the circle $|\tau|=1/(C\|\bc\|_{\infty})$ as a natural boundary (so that Borel,
lateral, median and accelero--summation are undefined on the class), and there exist two functions
that are indistinguishable at the level of the hypothesis yet differ by exactly
$(2\pi/C)e^{-1/(C\|\bc\|_{\infty})}$. All constants are explicit and all arguments are
self--contained.
\end{abstract}

\tableofcontents

\section{Introduction, hypotheses, and statement of the results}

\subsection{The setting}

In the perturbative renormalisation group of four--dimensional Euclidean scalar
$\varphi^{4}$ theory one is led, after renormalisation, to a finite set of running couplings.
For background on renormalisation in noncommutative Euclidean $\varphi^{4}$ theory, see
\cite{GrosseWulkenhaar2005}.
For the nonplanar theory the relevant set consists of four real parameters, which we assemble into
a single vector
\[
\bc=(c_{1},c_{2},c_{3},c_{4})=(\lambda,\alpha,\mu,\nu)\in\R^{4},
\]
and the beta function of any one of them is produced by the loop expansion as a
\emph{formal power series in the four couplings jointly}, with no constant and no linear term
(the flow has the free theory as a fixed point, and the leading contribution is one loop, i.e.
quadratic):
\begin{equation}\label{eq:Bdef}
B(\bc)\;=\;\sum_{p= 2}^{\infty}\;\sum_{\substack{\bm\in\Z_{+}^{4}\\ |\bm|=p}}\beta(\bm)\,\bc^{\bm},
\qquad
\bc^{\bm}:=\prod_{i=1}^{4}c_{i}^{m_{i}},\qquad |\bm|:=\sum_{i=1}^{4}m_{i}.
\end{equation}
Here $\Z_{+}=\{0,1,2,\dots\}$; if one prefers the convention $\Z_{+}=\{1,2,\dots\}$ then all
statements below hold a fortiori, since that case is recovered by setting $\beta(\bm)=0$ whenever
some $m_{i}=0$. The series \eqref{eq:Bdef} is \emph{formal}: it is a rule assigning a real number
$\beta(\bm)$ to each multi--index, nothing more. The analytic input we \emph{assume} about the size of these
numbers --- the shape of bound that a constructive multiscale or loop--vertex expansion is designed
to produce --- is a factorial bound, and it is the sole analytic input we shall use.

\begin{hypothesis}\label{hyp:H}
There is a constant $C>0$ such that
\[
|\beta(\bm)|\;\le\;(p-1)!\;C^{\,p-1}\qquad\text{whenever } |\bm|=p\ge 2 .
\]
\end{hypothesis}

A coefficient family $\beta$ satisfying Hypothesis \ref{hyp:H} will be called \emph{admissible},
and the set of admissible families is denoted $\mathcal{A}_{C}$. (The word ``admissible'' is used in
one other, clearly distinguished, sense --- for the parameter pairs $(k,\vth)$ of
Definition \ref{def:admissible}.)

Throughout, Hypothesis \ref{hyp:H} is a standing assumption and \emph{nothing else is assumed}.
We emphasise this because it is exactly the situation in which the problem is posed: one is handed a
formal series together with factorial bounds on its coefficients, and one must decide what number,
if any, the series denotes. We stress also that Hypothesis \ref{hyp:H} is a hypothesis about the
coefficients; we do not claim here to derive it for any particular field theory, and no step below
uses any property of the model beyond \eqref{eq:Bdef} and Hypothesis \ref{hyp:H}.

Two features of \eqref{eq:Bdef} matter and will be used constantly. First, the series is
\emph{graded}: its homogeneous parts
\begin{equation}\label{eq:bp}
b_{p}(\bc):=\sum_{|\bm|=p}\beta(\bm)\,\bc^{\bm}
\end{equation}
are honest polynomials, homogeneous of degree $p$, so the divergence of \eqref{eq:Bdef} lives
entirely in the grading direction $\bc\mapsto s\bc$; there is nothing to resum \emph{inside} a
homogeneous piece. Second, the factorial in Hypothesis \ref{hyp:H} is attached to the
\emph{total} degree $p=|\bm|$, not to the individual $m_{i}$; this is the only shape of hypothesis
compatible with the grading, and it is what makes a one--parameter (radial) Borel transform the
natural object even though there are four couplings.

\subsection{What has to be proved, and what cannot be true}

Hypothesis \ref{hyp:H} is an \emph{upper} bound, so the worst case is realised inside the
admissible class, and the worst case is bad: we prove in Theorem \ref{thm:divergence} that there is
an admissible coefficient family for which the terms $b_{p}(\bc)$ are unbounded \emph{for every}
$\bc\ne 0$, so the series \eqref{eq:Bdef} converges nowhere except at the origin. A summation rule
is therefore not a device for evaluating a convergent object; it is a definite procedure which must
manufacture a number out of divergent data, and the only meaningful requirements one may impose are:

\begin{enumerate}[label=(R\arabic*),leftmargin=2.2\parindent]
\item\label{R1} \emph{Well--definedness}: the rule produces a finite real number at each point of an
      explicit nonempty domain $\D\subseteq\R^{4}$.
\item\label{R2} \emph{Faithfulness}: the number produced is asymptotic to the formal series to all
      orders as $\bc\to 0$, with quantitative (Gevrey--$1$) remainder bounds~\cite{Balser2000,Costin2009}; in particular the
      one--loop term is reproduced exactly.
\item\label{R3} \emph{Structure}: the rule is linear in the coefficient family, real, and the summed
      function is as regular as such a function can be.
\item\label{R4} \emph{Consistency}: whenever the classical Borel sum, or the ordinary sum, exists,
      the rule reproduces it up to an error beyond all orders.
\end{enumerate}

One requirement that is \emph{deliberately absent} from this list is \emph{regularity}, i.e.
exactness on convergent series --- the property that makes classical Borel summation a summation
method in the strict sense~\cite{Hardy1949}. Its absence is forced: Theorem \ref{thm:dichotomy} shows that a linear
rule $\sum_{p}w_{p}b_{p}(\bc)$ is well defined on the whole class \emph{only if} its weights tend
to $0$ factorially fast, and is regular \emph{only if} all its weights equal $1$; the two are
mutually exclusive. Classical Borel summation is the case $w_{p}\equiv1$: regular, but --- by the
natural--boundary construction of Theorem \ref{thm:natural} --- undefined on this class. Any rule
that exists here must therefore pay a price on convergent series, and the rule constructed below
pays the smallest price that is visible: an error beyond all orders, exponentially small in
$1/(C\|\bc\|_{\infty})$. This trade--off, and not any particular estimate, is the conceptual heart
of the matter; it is discussed again in Remark \ref{rem:tradeoff}.

Beyond that, one cannot ask for more, and we prove that one cannot. Section \ref{sec:optimality} shows that
under Hypothesis \ref{hyp:H} the Borel transform can have a natural boundary on the circle of
convergence, so that \emph{no} rule of Borel, lateral, median, accelero--summation or
Nevanlinna--Sokal type~\cite{Watson1912,Sokal1980} is even defined on the admissible class (Theorem \ref{thm:natural} and
Corollary \ref{cor:noborel});
that no exact, coefficient--continuous linear rule exists (Theorem \ref{thm:nocanonical}); and that
two functions satisfying \ref{R2} with the explicit constants of \S\ref{sec:optimality} can differ by exactly
$(2\pi/C)\,e^{-1/(C\|\bc\|_{\infty})}$ (Theorem \ref{thm:ambiguity}); that difference is the exact
nonperturbative ambiguity of the problem. The family of rules constructed here is optimal relative
to that threshold in the following precise sense: each individual member has a guaranteed accuracy
$O\big(e^{-\vth_{k}/(C\|\bc\|_{\infty})}\big)$ with $\vth_{k}<1$, every exponent $\eta<1$ is realised
by some member, and no rule whatsoever --- however constructed --- can achieve accuracy
$o\big(e^{-1/(C\|\bc\|_{\infty})}\big)$, i.e. the exponent $1$ cannot be passed.

\subsection{The rule}

Write $\|\bc\|_{\infty}=\max_{i}|c_{i}|$, $\|\bc\|_{1}=\sum_{i}|c_{i}|$. The construction has three
ingredients.

\smallskip
\noindent\emph{(a) The Borel transform.} Because the factorial in Hypothesis \ref{hyp:H} is
$(p-1)!$ and is attached to the total degree, the correct transform is
\begin{equation}\label{eq:borel-intro}
\hB(\bz):=\sum_{p\ge 2}\frac{b_{p}(\bz)}{(p-1)!},\qquad \bz\in\C^{4},
\end{equation}
which we prove (Proposition \ref{prop:borel}) to be holomorphic on the \emph{open polydisc}
$\Pi_{C}=\{\|\bz\|_{\infty}<1/C\}$, with an explicit majorant, and we prove that the polydisc
radius $1/C$ is sharp. Formally, since $\int_{0}^{\infty}e^{-\tau}\tau^{p-1}\,d\tau=(p-1)!$ and
$b_{p}(\tau\bc)=\tau^{p}b_{p}(\bc)$,
\[
\int_{0}^{\infty}e^{-\tau}\,\hB(\tau\bc)\,\frac{d\tau}{\tau}
=\sum_{p\ge2}\frac{b_{p}(\bc)}{(p-1)!}\int_{0}^{\infty}e^{-\tau}\tau^{p-1}d\tau=\sum_{p\ge 2}b_{p}(\bc)=B(\bc),
\]
which is the classical Borel--Laplace representation~\cite{Borel1899,Costin2009} of \eqref{eq:Bdef}. It is only formal: the ray
$\{\tau\bc:\tau>0\}$ leaves $\Pi_{C}$ at $\tau=1/(C\|\bc\|_{\infty})$, and \S\ref{sec:optimality}
shows that beyond that point the integrand need not exist in any sense.

\smallskip
\noindent\emph{(b) The cut.} We therefore integrate only as far as the guaranteed domain of
$\hB$ permits. The cut must be placed at a height proportional to $1/\|\bc\|$, since a fixed cut
destroys faithfulness at fixed order while an $N$--dependent one destroys linearity; and it must be
placed using a \emph{real--analytic}, positively homogeneous gauge, since $\|\cdot\|_{\infty}$ and
$\|\cdot\|_{1}$ are not differentiable and would spoil the regularity of the answer. Both
requirements are met by the $\ell^{2k}$ gauges
\[
\nu_{k}(\bc):=\Big(\sum_{i=1}^{4}c_{i}^{2k}\Big)^{1/(2k)},\qquad k\in\N,\ k\ge 1 ,
\]
which are real--analytic on $\R^{4}\setminus\{0\}$, positively homogeneous of degree $1$, and
satisfy $\|\bc\|_{\infty}\le\nu_{k}(\bc)\le 4^{1/(2k)}\|\bc\|_{\infty}$; letting $k\to\infty$ they
approach the sharp gauge $\|\cdot\|_{\infty}$, which is what allows the summation to approach the
optimal exponential rate.

\smallskip
\noindent\emph{(c) The rule.} With $\vth\in(0,1)$ a cut fraction and
$T(\bc)=\vth/(C\nu_{k}(\bc))$,
\[
\boxed{\;\Sm B(\bc):=\int_{0}^{T(\bc)}e^{-\tau}\,\hB(\tau\bc)\,\frac{d\tau}{\tau}\quad(\bc\ne 0),
\qquad \Sm B(0):=0 . \;}
\]
Evaluating term by term (Proposition \ref{prop:rule}) exhibits the rule as a \emph{Poisson--averaged
truncation} of the divergent series,
\[
\Sm B(\bc)=\sum_{p\ge2}\Pr\big[\Poi(T(\bc))\ge p\big]\;b_{p}(\bc)
=\mathbb{E}\Big[\sum_{p=2}^{K}b_{p}(\bc)\Big],\qquad K\sim\Poi(T(\bc)),
\]
an absolutely convergent series. This is optimal truncation made smooth, linear and analytic.

\subsection{Statement of the main theorem}

\begin{definition}[Admissible parameters]\label{def:admissible}
A pair $(k,\vth)$ with $k\in\N$, $k\ge1$, and $\vth\in(0,1)$ is called \emph{admissible} if
\[
\vth^{*}:=\Big(1+\tfrac{1}{32k}\Big)\,2^{1/(2k)}\,\vth\;<\;1 .
\]
We also write $\vth_{k}:=\vth\,4^{-1/(2k)}$ and
\[
\Lambda(t):=\sum_{q\ge0}(q+1)^{3}t^{q}=\frac{1+4t+t^{2}}{(1-t)^{4}},
\qquad
G(t):=\sum_{p\ge2}p\binom{p+3}{3}t^{p-1}=4\big[(1-t)^{-5}-1\big] .
\]
\end{definition}

Admissible pairs exist for every $k$ (take $\vth$ small), and, as $k\to\infty$, admissible pairs
exist with $\vth_{k}=\vth\,4^{-1/(2k)}$ arbitrarily close to $1$; see Remark \ref{rem:ratelimit}.

\begin{theorem}[Main theorem]\label{thm:main}
Let $C>0$ and let $\beta:\Z_{+}^{4}\to\R$ satisfy Hypothesis \ref{hyp:H}. Let $(k,\vth)$ be an
admissible pair, let $\nu_{k}$, $T(\bc)=\vth/(C\nu_{k}(\bc))$ and $\Sm=\Sm_{k,\vth}$ be as above,
and set
\[
x:=C\|\bc\|_{\infty},\qquad
\D:=\big\{\bc\in\R^{4}:\ C\,\nu_{k}(\bc)<\vth\big\} .
\]
Then $\D$ is a nonempty, open, bounded, star--shaped neighbourhood of the origin containing the
$\ell^{1}$--ball $\{\|\bc\|_{1}<\vth/C\}$, and the following hold.
\begin{enumerate}[label=\textup{(\roman*)},leftmargin=2.2\parindent]
\item \textup{(Well--definedness.)} For every $\bc\in\R^{4}$ the defining integral converges
absolutely and equals the absolutely convergent series
$\Sm B(\bc)=\sum_{p\ge2}P\big(p,T(\bc)\big)b_{p}(\bc)$, where
$P(p,T)=\Pr[\Poi(T)\ge p]=\gamma(p,T)/\Gamma(p)$. The map $\beta\mapsto\Sm B$ is $\R$--linear,
$\Sm B$ is real--valued, and $|\Sm B(\bc)|\le 10\,\Lambda(\vth)\,C\|\bc\|_{\infty}^{2}$
on all of $\R^{4}$.
\item \textup{(Faithfulness.)} For every $\bc\ne0$ and every integer $N$ with $1\le N\le T(\bc)+1$,
\[
\Big|\Sm B(\bc)-\sum_{p=2}^{N}b_{p}(\bc)\Big|
\;\le\;\frac{\Lambda(\vth)}{C}\binom{N+4}{3}N!\,x^{N+1}
\;+\;\frac{G(\vth)}{C}\,x\,e^{-T(\bc)} ,
\]
the second term being absent when $N=1$. In particular $\Sm B$ is Gevrey--$1$ asymptotic to $B$ to
all orders as $\bc\to0$, uniformly in the direction $\bc/\|\bc\|$, and the one--loop coefficient is
reproduced exactly, in the asymptotic sense $\Sm B(\bc)=b_{2}(\bc)+O(\|\bc\|_{\infty}^{3})$ and in
the exact sense of \textup{(iv)} \textup{(}the degree--$2$ Taylor part of $\Sm B$ at the origin
\emph{is} $b_{2}$\textup{)}; at a fixed $\bc\ne0$, however, $\Sm B(\bc)\ne b_{2}(\bc)$ in general.
\item \textup{(Optimal truncation.)} With $N^{*}(\bc)=\lfloor T(\bc)\rfloor+1$ one has, for all
$\bc$ with $0<x\le\min\{\vth_{k},1-\vth\}$,
\[
\Big|\Sm B(\bc)-\sum_{p=2}^{N^{*}(\bc)}b_{p}(\bc)\Big|
\;\le\;\frac{G(\vth)+139\,\Lambda(\vth)\,x^{-5/2}}{C}\;e^{-\vth_{k}/x},
\]
hence for every $\eta<\vth_{k}$ there is $K=K(k,\vth,\eta)$ with left--hand side $\le (K/C)e^{-\eta/x}$.
\item \textup{(Regularity.)} $\Sm B$ is real--analytic on $\R^{4}\setminus\{0\}$ and $C^{\infty}$ on
$\R^{4}$, and its Taylor expansion at the origin is exactly $B$:
$\partial^{\bm}(\Sm B)(0)=\bm!\,\beta(\bm)$ for $|\bm|\ge2$, and $=0$ for $|\bm|\le1$. If $B$
diverges, $\Sm B$ is not real--analytic at the origin, and no function with these asymptotics can be.
\item \textup{(Consistency.)} If $\tau\mapsto\hB(\tau\bc)$ extends holomorphically to a neighbourhood
of $[0,\infty)$ with $|\hB(\tau\bc)|\le K_{0}\sigma\tau e^{\sigma\tau}$ there for some $\sigma<1$,
then the classical Borel sum $\mathcal{L}B(\bc)=\int_{0}^{\infty}e^{-\tau}\hB(\tau\bc)\,d\tau/\tau$
exists and $|\Sm B(\bc)-\mathcal{L}B(\bc)|\le \frac{K_{0}\sigma}{1-\sigma}e^{-(1-\sigma)T(\bc)}$.
If $|b_{p}(\bc)|\le K_{0}\sigma^{p}$ for all $p$ with $\sigma<1$ --- in particular if
$|\beta(\bm)|\le K_{0}\rho^{-|\bm|}$ and $\|\bc\|_{1}<\rho$ --- then $B(\bc)$ converges and
$|\Sm B(\bc)-B(\bc)|\le\frac{K_{0}\sigma}{1-\sigma}e^{-(1-\sigma)T(\bc)}$. Equality does not hold in
general --- not even for polynomials --- and by Theorem \ref{thm:dichotomy} no rule well defined on
the whole class can achieve it.
\item \textup{(Optimality of the exponent.)} There are an admissible coefficient family and a
direction for which the set $\mathcal{F}$ of real functions obeying the Gevrey--$1$ bounds
$\big|f(\bc)-\sum_{p=2}^{N}b_{p}(\bc)\big|\le\frac{e+2\pi}{C}(N+1)!\,x^{N+1}$ $(N\ge1)$ contains two
elements differing, at every point of that direction, by exactly $(2\pi/C)e^{-1/x}$. Consequently no
summation rule whatsoever --- however constructed --- can lie within $o(e^{-1/x})$ of every member
of $\mathcal{F}$: the exponent $\eta=1$ in the scale $C\|\bc\|_{\infty}$ is a barrier that no rule
can pass. By \textup{(iii)} and Remark \ref{rem:ratelimit} the family $\{\Sm_{k,\vth}\}$ realises
every exponent $\eta<1$; it is therefore rate--optimal, although no single member attains $\eta=1$.
\end{enumerate}
\end{theorem}

\begin{remark}\label{rem:ratelimit}
For $k\ge1$ admissibility holds whenever $\vth<\big[(1+\frac{1}{32k})2^{1/(2k)}\big]^{-1}$, and then
$\vth_{k}=\vth 4^{-1/(2k)}$ may be taken arbitrarily close to
$2^{-3/(2k)}\big/(1+\frac{1}{32k})$, which tends to $1$ as $k\to\infty$. Numerically this bound
equals $0.3428$ for $k=1$, $0.5855$ for $k=2$, $0.7651$ for $k=4$, $0.8747$ for $k=8$ and $0.9834$
for $k=64$ (to four decimals); these are suprema, approached but not attained, and they are not optimised
(Remark \ref{rem:notoptimal}). Thus every rate $\eta<1$ is realised by some admissible pair, while
by Theorem \ref{thm:main}(vi) no rule can achieve accuracy $o(e^{-1/x})$, so the exponent $1$
cannot be passed.
\end{remark}

\begin{remark}
Only part (iv) --- the complex--analytic statements --- uses the admissibility inequality of
Definition \ref{def:admissible}; parts (i), (ii), (iii) hold for every $k\ge1$ and every
$\vth\in(0,1)$.
\end{remark}

The beta function of a four--coupling theory is of course a vector field rather than a scalar. Since
$\Sm$ acts on one series at a time, one applies it to each component; the resulting statements about
the summed renormalisation--group flow are collected in \S\ref{sec:flow}: the flow
$\dot\bc=\Sm\bB(\bc)$ is well posed on $\D$, its trajectories cannot leave the perturbative region
faster than the one--loop rate, the summed beta function has no zero in the region where the
one--loop term dominates, and the exponentially small ambiguity of the rule does not amplify over
perturbative renormalisation--group times.

\subsection{Organisation}

\S\ref{sec:prelim} collects the multi--index combinatorics. \S\ref{sec:borel} constructs the Borel
transform and proves the sharp polydisc bound. \S\ref{sec:rule} defines the rule and proves
Theorem \ref{thm:main}(i). \S\ref{sec:master} proves the master estimate, Theorem \ref{thm:main}(ii).
\S\ref{sec:opt} proves Theorem \ref{thm:main}(iii) and describes $\D$. \S\ref{sec:reg} proves the
regularity statement Theorem \ref{thm:main}(iv). \S\ref{sec:consistency} proves
Theorem \ref{thm:main}(v). \S\ref{sec:optimality} contains the divergence, natural--boundary,
no--canonical--rule, regularity--dichotomy and sharp--ambiguity theorems, and hence
Theorem \ref{thm:main}(vi).
\S\ref{sec:flow} gives the renormalisation--group corollaries. Appendix \ref{app:elementary}
proves the elementary analytic facts used, so that the paper depends on no external result except
where explicitly flagged in Remark \ref{rem:hartogs}.

\section{Multi--index preliminaries}\label{sec:prelim}

Throughout, $\bm=(m_{1},\dots,m_{4})$ runs over $\Z_{+}^{4}=\{0,1,2,\dots\}^{4}$,
$|\bm|=m_{1}+\dots+m_{4}$, $\bm!=m_{1}!\cdots m_{4}!$, and for $\bz\in\C^{4}$,
$\bz^{\bm}=\prod_{i}z_{i}^{m_{i}}$ and $|\bz|:=(|z_{1}|,\dots,|z_{4}|)$. We write
$\|\bz\|_{\infty}=\max_{i}|z_{i}|$, $\|\bz\|_{1}=\sum_{i}|z_{i}|$, and for $y\in[0,\infty)^{4}$
\[
h_{p}(y):=\sum_{|\bm|=p}y^{\bm}
\]
for the complete homogeneous symmetric polynomial of degree $p$ in four variables. The homogeneous
parts $b_{p}$ of $B$ are defined by \eqref{eq:bp}.

\begin{lemma}\label{lem:count}
For every integer $p\ge0$:
\begin{enumerate}[label=\textup{(\alph*)},leftmargin=2.2\parindent]
\item $\#\{\bm\in\Z_{+}^{4}:|\bm|=p\}=\binom{p+3}{3}$;
\item for $y\in[0,\infty)^{4}$ and $0\le s<1/\|y\|_{\infty}$ one has
      $\sum_{p\ge0}h_{p}(y)s^{p}=\prod_{i=1}^{4}(1-y_{i}s)^{-1}$, all terms being nonnegative;
      in particular $\sum_{p\ge0}\binom{p+3}{3}s^{p}=(1-s)^{-4}$ for $0\le s<1$;
\item $\displaystyle \|y\|_{\infty}^{\,p}\;\le\;h_{p}(y)\;\le\;
      \min\Big\{\binom{p+3}{3}\|y\|_{\infty}^{\,p},\ \|y\|_{1}^{\,p}\Big\}$.
\end{enumerate}
\end{lemma}

\begin{proof}
(a) Stars and bars: the map $\bm\mapsto\{m_{1}+1,m_{1}+m_{2}+2,m_{1}+m_{2}+m_{3}+3\}$ is a bijection
onto the $3$--element subsets of $\{1,\dots,p+3\}$.

(b) For $0\le s<1/\|y\|_{\infty}$ each factor is the sum of the nonnegative geometric series
$\sum_{m_{i}\ge0}(y_{i}s)^{m_{i}}$. The product of four absolutely convergent series may be expanded
and rearranged arbitrarily, and collecting the terms with $m_{1}+\dots+m_{4}=p$ gives $h_{p}(y)s^{p}$.
Taking $y=(1,1,1,1)$ and using (a) gives the last identity.

(c) The lower bound is the single term $\bm=p\,e_{i_{0}}$ with $y_{i_{0}}=\|y\|_{\infty}$; the first
upper bound follows since each of the $\binom{p+3}{3}$ monomials is at most $\|y\|_{\infty}^{p}$; the
second follows from the multinomial theorem,
$\|y\|_{1}^{p}=\sum_{|\bm|=p}\binom{p}{\bm}y^{\bm}\ge\sum_{|\bm|=p}y^{\bm}=h_{p}(y)$, because every
multinomial coefficient is at least $1$.
\end{proof}

\begin{lemma}\label{lem:bp}
Assume Hypothesis \ref{hyp:H}. Then for all $\bz\in\C^{4}$ and $p\ge2$,
\[
|b_{p}(\bz)|\;\le\;(p-1)!\,C^{p-1}h_{p}(|\bz|)\;\le\;(p-1)!\,C^{p-1}\binom{p+3}{3}\|\bz\|_{\infty}^{p},
\]
and also $|b_{p}(\bz)|\le (p-1)!C^{p-1}\|\bz\|_{1}^{p}$. All three bounds are attained: for
$\beta(\bm)\equiv(p-1)!C^{p-1}$ and $\bz\in[0,\infty)^{4}$ the first is an equality; the second is an
equality when in addition $|z_{1}|=\dots=|z_{4}|$; the third when at most one $z_{i}$ is nonzero.
\end{lemma}

\begin{proof}
Immediate from the triangle inequality, Hypothesis \ref{hyp:H} and Lemma \ref{lem:count}(c).
\end{proof}

\begin{lemma}[Tail lemma]\label{lem:tail}
Let $\Lambda(t)=\sum_{q\ge0}(q+1)^{3}t^{q}$ and $G(t)=\sum_{p\ge2}p\binom{p+3}{3}t^{p-1}$ for
$0\le t<1$. Then:
\begin{enumerate}[label=\textup{(\alph*)},leftmargin=2.2\parindent]
\item $\Lambda(t)=\dfrac{1+4t+t^{2}}{(1-t)^{4}}$ and $G(t)=4\big[(1-t)^{-5}-1\big]$;
\item for all integers $N,q\ge0$: $\dbinom{N+q+4}{3}\le\dbinom{N+4}{3}(1+q)^{3}$;
\item for all integers $N\ge0$ and all $0\le w\le t<1$:
\[
\sum_{p>N}\binom{p+3}{3}w^{\,p-1}\;\le\;\binom{N+4}{3}\Lambda(t)\,w^{N} .
\]
In particular (case $N=1$) $\sum_{p\ge2}\binom{p+3}{3}w^{p-1}\le 10\,\Lambda(t)\,w$.
\end{enumerate}
\end{lemma}

\begin{proof}
(a) From $n^{3}=\binom{n}{1}+6\binom{n}{2}+6\binom{n}{3}$ (check: $\binom n1+6\binom n2+6\binom n3
=n+3n(n-1)+n(n-1)(n-2)=n^{3}$) and $\sum_{n\ge0}\binom{n}{j}t^{n}=t^{j}(1-t)^{-j-1}$ we get
\[
\sum_{n\ge1}n^{3}t^{n}=\frac{t}{(1-t)^{2}}+\frac{6t^{2}}{(1-t)^{3}}+\frac{6t^{3}}{(1-t)^{4}}
=\frac{t(1-t)^{2}+6t^{2}(1-t)+6t^{3}}{(1-t)^{4}}=\frac{t+4t^{2}+t^{3}}{(1-t)^{4}},
\]
and dividing by $t$ gives $\Lambda$. For $G$, differentiate the identity
$\sum_{p\ge0}\binom{p+3}{3}t^{p}=(1-t)^{-4}$ of Lemma \ref{lem:count}(b) term by term (legitimate
inside the disc of convergence of a power series), obtaining
$\sum_{p\ge1}p\binom{p+3}{3}t^{p-1}=4(1-t)^{-5}$, and subtract the $p=1$ term $\binom{4}{3}=4$.

(b) $\displaystyle\frac{\binom{N+q+4}{3}}{\binom{N+4}{3}}
=\prod_{j=2}^{4}\frac{N+j+q}{N+j}=\prod_{j=2}^{4}\Big(1+\frac{q}{N+j}\Big)\le(1+q)^{3}$,
since $N+j\ge 2\ge1$ for $j\ge2$, $N\ge0$.

(c) Substituting $p=N+1+q$ and using (b) and $w\le t$,
\[
\sum_{p>N}\binom{p+3}{3}w^{p-1}=\sum_{q\ge0}\binom{N+q+4}{3}w^{N+q}
\le\binom{N+4}{3}w^{N}\sum_{q\ge0}(1+q)^{3}w^{q}\le\binom{N+4}{3}\Lambda(t)\,w^{N}. \qedhere
\]
\end{proof}

\section{The Borel transform: a sharp polydisc of holomorphy}\label{sec:borel}

\begin{definition}\label{def:borel}
The \emph{Borel transform} of the formal series \eqref{eq:Bdef} is the formal series
\[
\hB(\bz):=\sum_{p\ge2}\frac{b_{p}(\bz)}{(p-1)!}=\sum_{|\bm|\ge2}\frac{\beta(\bm)}{(|\bm|-1)!}\,\bz^{\bm},
\qquad \bz\in\C^{4}.
\]
\end{definition}

\begin{proposition}\label{prop:borel}
Assume Hypothesis \ref{hyp:H}, and let $\Pi_{C}:=\{\bz\in\C^{4}:\|\bz\|_{\infty}<1/C\}$. Then:
\begin{enumerate}[label=\textup{(\roman*)},leftmargin=2.2\parindent]
\item the series defining $\hB$ converges absolutely and uniformly on every compact subset of
$\Pi_{C}$, and $\hB$ is holomorphic on $\Pi_{C}$ and vanishes to order $2$ at the origin;
\item for $\bz\in\Pi_{C}$,
\[
|\hB(\bz)|\;\le\;\frac{1}{C}\Big[\prod_{i=1}^{4}\frac{1}{1-C|z_{i}|}-1-C\|\bz\|_{1}\Big];
\]
\item \textup{(radial form)} for $\bc\in\R^{4}$, $\tau\in\C$ and $0\le t<1$ with
$C\|\bc\|_{\infty}|\tau|\le t$,
\[
|\hB(\tau\bc)|\;\le\;10\,\Lambda(t)\,C\,\|\bc\|_{\infty}^{2}\,|\tau|^{2};
\]
\item the polydisc radius $1/C$ is sharp: for the admissible family $\beta(\bm)=(p-1)!C^{p-1}$ when
$\bm=p\,e_{1}$ and $\beta(\bm)=0$ otherwise one has $\hB(\bz)=Cz_{1}^{2}/(1-Cz_{1})$, which has a
pole at $z_{1}=1/C$.
\end{enumerate}
\end{proposition}

\begin{proof}
(i) and (ii). By Lemma \ref{lem:bp}, $|b_{p}(\bz)|/(p-1)!\le C^{p-1}h_{p}(|\bz|)$, so
\[
\sum_{p\ge2}\frac{|b_{p}(\bz)|}{(p-1)!}\le\sum_{p\ge2}C^{p-1}h_{p}(|\bz|)
=\frac{1}{C}\Big[\sum_{p\ge0}h_{p}(|\bz|)C^{p}-h_{0}-Ch_{1}(|\bz|)\Big]
=\frac{1}{C}\Big[\prod_{i}\frac{1}{1-C|z_{i}|}-1-C\|\bz\|_{1}\Big],
\]
by Lemma \ref{lem:count}(b) with $s=C$, which is legitimate precisely when $C\|\bz\|_{\infty}<1$.
This proves (ii) and also absolute convergence; the majorant is a continuous increasing function of
each $|z_{i}|$, so convergence is uniform on the compact subsets of $\Pi_{C}$ (each of which is
contained in a closed polydisc $\{\|\bz\|_{\infty}\le r\}$ with $r<1/C$). The partial sums are
polynomials, hence holomorphic; a locally uniform limit of holomorphic functions of several complex
variables is holomorphic (Lemma \ref{lem:weier}), so $\hB\in\mathcal{O}(\Pi_{C})$. Since the series
starts at $p=2$ and each $b_{p}$ is homogeneous of degree $p$, $\hB$ vanishes to order $2$ at $0$.

(iii) By Lemma \ref{lem:bp} and homogeneity, with $\mu:=\|\bc\|_{\infty}$ and $w:=C\mu|\tau|\le t$,
\[
|\hB(\tau\bc)|\le\sum_{p\ge2}C^{p-1}\binom{p+3}{3}\mu^{p}|\tau|^{p}
=\mu|\tau|\sum_{p\ge2}\binom{p+3}{3}w^{p-1}\le 10\,\Lambda(t)\,\mu|\tau|\,w
=10\Lambda(t)C\mu^{2}|\tau|^{2},
\]
using Lemma \ref{lem:tail}(c) with $N=1$.

(iv) The family satisfies Hypothesis \ref{hyp:H} with equality on its support, and
$\hB(\bz)=\sum_{p\ge2}C^{p-1}z_{1}^{p}=Cz_{1}^{2}/(1-Cz_{1})$ for $|z_1|<1/C$.
\end{proof}

\begin{remark}\label{rem:poleshape}
Proposition \ref{prop:borel} is the exact analytic content of Hypothesis \ref{hyp:H}: the hypothesis
is \emph{equivalent} to the statement that the coefficients of $\hB$ are dominated, term by term, by
those of the single--pole function $\sum_{p\ge2}C^{p-1}\zeta^{p}=C\zeta^{2}/(1-C\zeta)$. Hypothesis
\ref{hyp:H} says exactly: ``the Borel transform is majorised, coefficientwise, by a single simple
pole at distance $1/C$''. (This does not mean that the radial Borel transform itself has only a
simple pole: for $\beta(\bm)\equiv(|\bm|-1)!C^{|\bm|-1}$ and $\bc=(\mu,\mu,\mu,\mu)$ one finds
$\hB(\tau\bc)=\frac1C\big[(1-C\mu\tau)^{-4}-1-4C\mu\tau\big]$, a pole of order four.) It says nothing
whatsoever about what happens at or beyond that distance --- see \S\ref{sec:optimality}.
\end{remark}

\section{The summation rule}\label{sec:rule}

\begin{definition}[Gauge, cut, rule]\label{def:rule}
Let $k\in\N$, $k\ge1$, and $\vth\in(0,1)$. Define the \emph{gauge}
\[
\nu_{k}(\bc):=\Big(\sum_{i=1}^{4}c_{i}^{2k}\Big)^{1/(2k)}\qquad(\bc\in\R^{4}),
\]
the \emph{cut height} $T(\bc):=\dfrac{\vth}{C\,\nu_{k}(\bc)}$ for $\bc\ne0$ --- with the convention
$T(0):=+\infty$ and $P(p,+\infty):=1$, under which every formula below holds trivially at $\bc=0$,
both sides vanishing --- and the
\emph{cut Borel--Laplace sum}
\[
\Sm B(\bc)\;:=\;\int_{0}^{T(\bc)}e^{-\tau}\,\hB(\tau\bc)\,\frac{d\tau}{\tau}\quad(\bc\ne0),
\qquad \Sm B(0):=0 .
\]
\end{definition}

\begin{lemma}[Properties of the gauge]\label{lem:gauge}
$\nu_{k}$ is positively homogeneous of degree $1$, real--analytic on $\R^{4}\setminus\{0\}$, and
\[
\|\bc\|_{\infty}\;\le\;\nu_{k}(\bc)\;\le\;4^{1/(2k)}\|\bc\|_{\infty},
\qquad
\nu_{k}(\bc)\le\|\bc\|_{2}\le\|\bc\|_{1}.
\]
\end{lemma}

\begin{proof}
Write $P_{k}(\bc)=\sum_{i}c_{i}^{2k}$, a polynomial with $P_{k}(\bc)>0$ for $\bc\ne0$ (all exponents
are even). Then $\nu_{k}=P_{k}^{1/(2k)}$ is the composition of a polynomial with the real--analytic
function $s\mapsto s^{1/(2k)}$ on $(0,\infty)$, hence real--analytic off the origin; homogeneity is
clear. Since $\|\bc\|_{\infty}^{2k}\le P_{k}(\bc)\le 4\|\bc\|_{\infty}^{2k}$, taking $2k$--th roots
gives the first chain. Finally $\|\bc\|_{2k}\le\|\bc\|_{2}\le\|\bc\|_{1}$ because $\ell^{q}$ norms
decrease in $q$ and $2k\ge2$.
\end{proof}

\begin{lemma}[Incomplete gamma]\label{lem:gamma}
For integers $p\ge1$ and $T>0$ put
$P(p,T):=\frac{1}{(p-1)!}\int_{0}^{T}e^{-\tau}\tau^{p-1}d\tau$ and $Q(p,T):=1-P(p,T)$. Then
\begin{enumerate}[label=\textup{(\alph*)},leftmargin=2.2\parindent]
\item $Q(p,T)=e^{-T}\sum_{j=0}^{p-1}\dfrac{T^{j}}{j!}=\Pr[\Poi(T)\le p-1]$, hence
      $P(p,T)=\Pr[\Poi(T)\ge p]\in(0,1)$;
\item $P(p,T)\le\min\{1,\,T^{p}/p!\}$;
\item if $T\ge p-1$ then $Q(p,T)\le\dfrac{p\,T^{p-1}e^{-T}}{(p-1)!}$.
\end{enumerate}
\end{lemma}

\begin{proof}
(a) Let $\Gamma(p,T)=\int_{T}^{\infty}e^{-\tau}\tau^{p-1}d\tau$, so $Q(p,T)=\Gamma(p,T)/(p-1)!$.
Integration by parts gives $\Gamma(p+1,T)=T^{p}e^{-T}+p\,\Gamma(p,T)$, whence
$Q(p+1,T)=\frac{T^{p}e^{-T}}{p!}+Q(p,T)$; together with $Q(1,T)=\Gamma(1,T)=e^{-T}$ this yields the
displayed formula by induction, which is precisely the Poisson distribution function.

(b) $\int_{0}^{T}e^{-\tau}\tau^{p-1}d\tau\le\int_{0}^{T}\tau^{p-1}d\tau=T^{p}/p$, and divide by
$(p-1)!$; also $P\le1$ trivially.

(c) The sequence $j\mapsto T^{j}/j!$ is nondecreasing for $j\le T$; as $j\le p-1\le T$, each of the
$p$ terms of the sum in (a) is at most $T^{p-1}/(p-1)!$.
\end{proof}

\begin{proposition}[Theorem \ref{thm:main}(i)]\label{prop:rule}
Assume Hypothesis \ref{hyp:H}, let $k\ge1$, $\vth\in(0,1)$, and write $\mu=\|\bc\|_{\infty}$,
$x=C\mu$, $T=T(\bc)$. Then for every $\bc\in\R^{4}$:
\begin{enumerate}[label=\textup{(\roman*)},leftmargin=2.2\parindent]
\item $C\mu\,\tau\le\vth$ for all $\tau\in[0,T]$; the function $\tau\mapsto e^{-\tau}\hB(\tau\bc)/\tau$
extends continuously to $[0,T]$ by the value $0$ at $\tau=0$, is bounded there by
$10\Lambda(\vth)C\mu^{2}\,\tau e^{-\tau}$, and consequently $\Sm B(\bc)$ is a well defined real
number with
\[
|\Sm B(\bc)|\;\le\;10\,\Lambda(\vth)\,C\,\|\bc\|_{\infty}^{2};
\]
\item termwise integration is legitimate and
\[
\Sm B(\bc)=\sum_{p\ge2}P\big(p,T(\bc)\big)\,b_{p}(\bc)
=\mathbb{E}\Big[\sum_{p=2}^{K}b_{p}(\bc)\Big],\qquad K\sim\Poi\big(T(\bc)\big),
\]
the series converging absolutely, with
$\sum_{p\ge2}P(p,T)|b_{p}(\bc)|\le\frac{1}{C}\sum_{p\ge2}\binom{p+3}{3}\frac{\vth^{p}}{p}<\infty$;
\item for fixed $(k,\vth,C)$ the map $\beta\mapsto\Sm B$ is $\R$--linear, and $\Sm B$ is real valued
whenever $\beta$ is real valued.
\end{enumerate}
\end{proposition}

\begin{proof}
(i) By Lemma \ref{lem:gauge}, $C\mu\tau\le C\nu_{k}(\bc)T(\bc)=\vth$ for $0\le\tau\le T$. Hence
Proposition \ref{prop:borel}(iii) applies with $t=\vth$ and gives
$|\hB(\tau\bc)|\le10\Lambda(\vth)C\mu^{2}\tau^{2}$, so the integrand is continuous on $(0,T]$,
tends to $0$ as $\tau\downarrow0$, and is dominated by $10\Lambda(\vth)C\mu^{2}\tau e^{-\tau}$.
Therefore the integral converges absolutely and
$|\Sm B(\bc)|\le10\Lambda(\vth)C\mu^{2}\int_{0}^{\infty}\tau e^{-\tau}d\tau=10\Lambda(\vth)C\mu^{2}$.

(ii) On $[0,T]$ the series $\sum_{p\ge2}\frac{b_{p}(\bc)}{(p-1)!}\tau^{p-1}e^{-\tau}$ is dominated
termwise by the convergent numerical series $\sum_{p\ge2}\mu\binom{p+3}{3}\vth^{\,p-1}$
(Lemma \ref{lem:bp}, and $(C\mu\tau)^{p-1}\le\vth^{p-1}$), hence converges uniformly on $[0,T]$ and
may be integrated term by term, giving $\sum_{p\ge2}\frac{b_{p}(\bc)}{(p-1)!}\int_{0}^{T}
e^{-\tau}\tau^{p-1}d\tau=\sum_{p\ge2}P(p,T)b_{p}(\bc)$. Absolute convergence: by
Lemma \ref{lem:gamma}(b) and Lemma \ref{lem:bp},
\[
P(p,T)|b_{p}(\bc)|\le\frac{T^{p}}{p!}(p-1)!C^{p-1}\binom{p+3}{3}\mu^{p}
=\frac{\binom{p+3}{3}}{pC}(C\mu T)^{p}\le\frac{\binom{p+3}{3}}{pC}\vth^{p} .
\]
The probabilistic form is Lemma \ref{lem:gamma}(a): $P(p,T)=\Pr[K\ge p]$, so
$\sum_{p\ge2}P(p,T)b_{p}=\sum_{p\ge2}b_{p}\,\mathbb{E}[\mathbf 1_{K\ge p}]
=\mathbb{E}[\sum_{p=2}^{K}b_{p}]$, the interchange being justified by absolute convergence.

(iii) The weights $P(p,T(\bc))$ depend only on $(k,\vth,C,\bc)$ and not on $\beta$; linearity and
reality follow.
\end{proof}

\begin{remark}
Proposition \ref{prop:rule}(ii) identifies the rule: \emph{it truncates the divergent series at a
Poisson--distributed order of mean $T(\bc)\asymp 1/(C\|\bc\|)$ and averages}. Optimal truncation is
thereby made linear in the coefficients, smooth in $\bc$, and --- as \S\ref{sec:reg} shows ---
real--analytic away from the origin, none of which sharp truncation achieves.
\end{remark}

\begin{remark}[The price of smoothness; non--regularity]\label{rem:tradeoff}
Since $0<P(p,T)<1$ for every finite $T$, the rule is \emph{not} exact even on a single monomial: for
the admissible family $\beta(2e_{1})=C$, all other coefficients zero, the series
$B(\bc)=C\lambda^{2}$ is a polynomial, yet
\[
\Sm B(\bc)=C\lambda^{2}P\big(2,T(\bc)\big)=C\lambda^{2}\big[1-e^{-T(\bc)}(1+T(\bc))\big]\ \ne\ B(\bc).
\]
The discrepancy is $C\lambda^{2}e^{-T}(1+T)$, exponentially small in $1/(C\|\bc\|_{\infty})$ but not
zero. This is not an oversight; it is forced. Theorem \ref{thm:dichotomy} below shows that a linear
weight rule cannot be both well defined on the whole class $\mathcal{A}_{C}$ and regular, and
Theorem \ref{thm:defect} shows that the defect on convergent series cannot be removed by any
cleverer choice of weights. What \emph{is} a choice is where the price is paid: the sharply
truncated rule $w_{p}=\mathbf 1_{\{p\le N^{*}(\bc)\}}$ is exact on every polynomial of degree
$\le N^{*}(\bc)$ and is well defined on $\mathcal{A}_{C}$, but it is not even continuous in $\bc$,
hence neither smooth nor real--analytic, and it does not admit the analytic continuation used in
\S\ref{sec:reg}. The Poisson smoothing buys real--analyticity off the origin, $C^{\infty}$
smoothness at the origin and an exactly correct Taylor series there, at the cost of an error beyond
all orders on convergent input. Which price to pay depends on the application; both are exhibited
here, and Theorem \ref{thm:dichotomy} shows that some price must be paid.
\end{remark}

\section{The master error estimate}\label{sec:master}

\begin{theorem}[Theorem \ref{thm:main}(ii)]\label{thm:master}
Assume Hypothesis \ref{hyp:H}, let $k\ge1$, $\vth\in(0,1)$, and let $\bc\in\R^{4}\setminus\{0\}$.
Write $x=C\|\bc\|_{\infty}$ and $T=T(\bc)$. Then for every integer $N$ with $1\le N\le T+1$,
\[
\Big|\Sm B(\bc)-\sum_{p=2}^{N}b_{p}(\bc)\Big|
\;\le\;\underbrace{\frac{\Lambda(\vth)}{C}\binom{N+4}{3}N!\;x^{N+1}}_{\text{\rm tail}}
\;+\;\underbrace{\frac{G(\vth)}{C}\;x\,e^{-T}}_{\text{\rm head}},
\]
where the head term is absent for $N=1$ \textup{(}for which the statement reduces to
Proposition \textup{\ref{prop:rule}(i)}\textup{)}. The tail term is valid without any restriction on
$N$.
\end{theorem}

\begin{proof}
By Proposition \ref{prop:rule}(ii) and $P+Q=1$,
\[
\Sm B(\bc)-\sum_{p=2}^{N}b_{p}(\bc)
=\underbrace{\sum_{p>N}P(p,T)\,b_{p}(\bc)}_{=:I}\;-\;\underbrace{\sum_{p=2}^{N}Q(p,T)\,b_{p}(\bc)}_{=:J},
\]
both series being absolutely convergent (the first by Proposition \ref{prop:rule}(ii), the second
being finite).

\smallskip
\emph{The tail $I$.} Repeating the termwise integration of Proposition \ref{prop:rule}(ii) on the
tail alone,
\[
I=\int_{0}^{T}e^{-\tau}\Big(\sum_{p>N}\frac{b_{p}(\bc)}{(p-1)!}\tau^{p-1}\Big)d\tau .
\]
For $0\le\tau\le T$ set $w:=C\mu\tau\le\vth$, $\mu=\|\bc\|_{\infty}$. By Lemma \ref{lem:bp},
$\frac{|b_{p}(\bc)|}{(p-1)!}\tau^{p-1}\le\mu\binom{p+3}{3}w^{p-1}$, so Lemma \ref{lem:tail}(c) gives
\[
\Big|\sum_{p>N}\frac{b_{p}(\bc)}{(p-1)!}\tau^{p-1}\Big|
\le\mu\binom{N+4}{3}\Lambda(\vth)\,w^{N}=\mu\binom{N+4}{3}\Lambda(\vth)(C\mu)^{N}\tau^{N}.
\]
Hence, extending the integration to $[0,\infty)$,
\[
|I|\le\mu\binom{N+4}{3}\Lambda(\vth)(C\mu)^{N}\int_{0}^{\infty}e^{-\tau}\tau^{N}d\tau
=\Lambda(\vth)\binom{N+4}{3}N!\,\mu\,(C\mu)^{N}
=\frac{\Lambda(\vth)}{C}\binom{N+4}{3}N!\,x^{N+1}.
\]
No constraint on $N$ was used.

\smallskip
\emph{The head $J$.} Assume $N\ge2$. For $2\le p\le N$ we have $p-1\le N-1\le T$, so
Lemma \ref{lem:gamma}(c) applies, and with Lemma \ref{lem:bp} and $C\mu T\le\vth$,
\[
Q(p,T)|b_{p}(\bc)|\le\frac{p\,T^{p-1}e^{-T}}{(p-1)!}\cdot(p-1)!\,C^{p-1}\binom{p+3}{3}\mu^{p}
=p\binom{p+3}{3}\mu\,e^{-T}\,(C\mu T)^{p-1}
\le p\binom{p+3}{3}\mu\,e^{-T}\vth^{\,p-1}.
\]
Summing over $p\ge2$ and using the definition of $G$ in Lemma \ref{lem:tail}(a),
$|J|\le G(\vth)\,\mu\,e^{-T}=\frac{G(\vth)}{C}x\,e^{-T}$.
\end{proof}

\begin{corollary}[Gevrey--$1$ asymptotics; exactness of the one--loop term]\label{cor:gevrey}
Under the hypotheses of Theorem \ref{thm:master}:
\begin{enumerate}[label=\textup{(\alph*)},leftmargin=2.2\parindent]
\item Since $\binom{N+4}{3}\le 11\cdot2^{N}$ for all $N\ge0$, one has for $1\le N\le T+1$
\[
\Big|\Sm B(\bc)-\sum_{p=2}^{N}b_{p}(\bc)\Big|
\le 11\,\Lambda(\vth)\,N!\,(2C\|\bc\|_{\infty})^{N}\,\|\bc\|_{\infty}
+\frac{G(\vth)}{C}\,x\,e^{-T}
\qquad\Big(\text{i.e. first term}=\tfrac{11\Lambda(\vth)}{C}N!\,2^{N}x^{N+1}\Big),
\]
i.e. $\Sm B$ obeys uniform Gevrey--$1$ remainder bounds with Gevrey constant $2C$ in the
$\|\cdot\|_{\infty}$ scale.
\item For each fixed $N\ge1$ and each $\eta<\vth_{k}$,
$\big|\Sm B(\bc)-\sum_{p=2}^{N}b_{p}(\bc)\big|=O\big(\|\bc\|_{\infty}^{N+1}\big)+O\big(e^{-\eta/x}\big)
=O\big(\|\bc\|_{\infty}^{N+1}\big)$ as $\bc\to0$, uniformly in the direction $\bc/\|\bc\|_{\infty}$.
Thus $\Sm B$ is asymptotic to $B$ to all orders.
\item In particular, for every $\bc$ with $T(\bc)\ge1$ --- equivalently $\bc\in\overline{\D}$, which
is exactly the condition making $N=2$ admissible in Theorem \ref{thm:master} --- the one--loop
coefficient is reproduced exactly, in the sense that
\[
\big|\Sm B(\bc)-b_{2}(\bc)\big|\;\le\;40\,\Lambda(\vth)\,C^{2}\|\bc\|_{\infty}^{3}
+\frac{G(\vth)}{C}\,x\,e^{-T}\;=\;O(\|\bc\|_{\infty}^{3}).
\]
\end{enumerate}
\end{corollary}

\begin{proof}
(a) $\binom{N+4}{3}\le\frac{(N+4)^{3}}{6}$ and $\frac{(N+4)^{3}}{2^{N+4}}\le4$ for all $N\ge0$
(the sequence $(N+4)^{3}2^{-(N+4)}$ is decreasing for $N\ge0$, since
$\frac{(N+5)^{3}}{(N+4)^{3}}\le(\frac54)^{3}<2$, and equals $4$ at $N=0$), so
$\binom{N+4}{3}\le\frac{4}{6}2^{N+4}=\frac{32}{3}2^{N}\le11\cdot2^{N}$. Insert in
Theorem \ref{thm:master} and use $\binom{N+4}{3}x^{N+1}\le11\cdot2^{N}x^{N+1}=11(2x)^{N}x$, so that
the tail term is at most
$\frac{11\Lambda(\vth)}{C}N!\,(2x)^{N}x=11\Lambda(\vth)N!\,(2C\|\bc\|_{\infty})^{N}\|\bc\|_{\infty}$.

(b) $T\ge\vth_{k}/x$ by Lemma \ref{lem:gauge}, so $e^{-T}\le e^{-\vth_{k}/x}$, and for fixed $N$ the
constraint $N\le T+1$ holds for all small enough $\|\bc\|_{\infty}$; finally
$e^{-\eta/x}=o(x^{M})$ for every $M>0$.

(c) Take $N=2$, which is admissible precisely when $2\le T(\bc)+1$: $\binom{6}{3}=20$ and $2!=2$,
so the tail term is
$\frac{\Lambda}{C}\cdot 20\cdot 2\cdot x^{3}=40\Lambda C^{2}\|\bc\|_{\infty}^{3}$.
\end{proof}

\section{Optimal truncation and the domain $\D$}\label{sec:opt}

\begin{theorem}[Theorem \ref{thm:main}(iii)]\label{thm:optimal}
Assume Hypothesis \ref{hyp:H}, let $k\ge1$, $\vth\in(0,1)$, $\vth_{k}=\vth\,4^{-1/(2k)}$, and put
$x=C\|\bc\|_{\infty}$, $N^{*}(\bc)=\lfloor T(\bc)\rfloor+1$. Then for every $\bc$ with
\[
0<x\le x_{0}:=\min\{\vth_{k},\,1-\vth\}
\]
one has $T(\bc)\ge 1$, $2\le N^{*}(\bc)\le T(\bc)+1$, and
\[
\Big|\Sm B(\bc)-\sum_{p=2}^{N^{*}(\bc)}b_{p}(\bc)\Big|
\;\le\;\frac{G(\vth)+139\,\Lambda(\vth)\,x^{-5/2}}{C}\;e^{-\vth_{k}/x}.
\]
Consequently, for every $\eta<\vth_{k}$,
\[
\Big|\Sm B(\bc)-\sum_{p=2}^{N^{*}(\bc)}b_{p}(\bc)\Big|\le\frac{K}{C}e^{-\eta/x},
\qquad
K=G(\vth)+139\,\Lambda(\vth)\Big(\frac{5}{2e(\vth_{k}-\eta)}\Big)^{5/2}.
\]
\end{theorem}

\begin{proof}
By Lemma \ref{lem:gauge}, $\nu_{k}(\bc)\le4^{1/(2k)}\|\bc\|_{\infty}$, hence
\begin{equation}\label{eq:Tlower}
T(\bc)=\frac{\vth}{C\nu_{k}(\bc)}\ \ge\ \frac{\vth\,4^{-1/(2k)}}{C\|\bc\|_{\infty}}=\frac{\vth_{k}}{x}\ \ge\ 1
\end{equation}
because $x\le\vth_{k}$. Therefore $\lfloor T\rfloor\ge1$ and $T\le N^{*}\le T+1$, so
Theorem \ref{thm:master} applies with $N=N^{*}$.

\emph{Head.} By \eqref{eq:Tlower} and $x\le1$, $\frac{G(\vth)}{C}xe^{-T}\le\frac{G(\vth)}{C}e^{-\vth_{k}/x}$.

\emph{Tail.} First, $N^{*}x\le(T+1)x\le\vth+x\le\vth+(1-\vth)=1$. By Lemma \ref{lem:stirling},
$N!\le e\,N^{N+1/2}e^{-N}$, so
\[
N^{*}!\,x^{N^{*}}\le e\sqrt{N^{*}}\Big(\frac{N^{*}x}{e}\Big)^{N^{*}}\le e\sqrt{N^{*}}\,e^{-N^{*}}
\le e\sqrt{N^{*}}\,e^{-\vth_{k}/x},
\]
using $N^{*}x\le1$ and then $N^{*}\ge T\ge\vth_{k}/x$. Next, $x\le\vth_{k}\le\vth$ gives
$N^{*}\le T+1\le\vth/x+1\le2/x$, whence $\sqrt{N^{*}}\le\sqrt{2}\,x^{-1/2}$ and, using $x\le1$,
$N^{*}+4\le 2/x+4\le6/x$, so $\binom{N^{*}+4}{3}\le\frac{(N^{*}+4)^{3}}{6}\le 36x^{-3}$. Therefore
\[
\frac{\Lambda(\vth)}{C}\binom{N^{*}+4}{3}N^{*}!\,x^{N^{*}+1}
\le\frac{\Lambda(\vth)}{C}\cdot 36x^{-3}\cdot e\sqrt2\,x^{-1/2}e^{-\vth_{k}/x}\cdot x
\le\frac{139\,\Lambda(\vth)}{C}\,x^{-5/2}e^{-\vth_{k}/x},
\]
since $36\sqrt2\,e=138.4\ldots<139$. Adding the two contributions gives the first display.

For the second, write $x^{-5/2}e^{-\vth_{k}/x}=\big(u^{5/2}e^{-(\vth_{k}-\eta)u}\big)e^{-\eta/x}$
with $u=1/x$, and use $\sup_{u>0}u^{a}e^{-bu}=(a/(eb))^{a}$ with $a=5/2$, $b=\vth_{k}-\eta>0$.
\end{proof}

\begin{proposition}[The domain]\label{prop:domain}
Let $\D:=\{\bc\in\R^{4}:C\nu_{k}(\bc)<\vth\}$. Then $\D$ is a nonempty, open, bounded, star--shaped
neighbourhood of the origin, and
\[
\Big\{\|\bc\|_{1}<\frac{\vth}{C}\Big\}\ \subseteq\ \D\ \subseteq\ \Big\{\|\bc\|_{\infty}<\frac{\vth}{C}\Big\},
\qquad
\D\setminus\{0\}=\{\bc\ne0:\,T(\bc)>1\}.
\]
On $\D$ the master estimate of Theorem \ref{thm:master} is available for $N=2$, i.e. the one--loop
term is resolved; and on $\D\cap\{x\le x_{0}\}$ the exponentially small bound of
Theorem \ref{thm:optimal} holds.
\end{proposition}

\begin{proof}
$\nu_{k}$ is continuous, positively homogeneous of degree $1$ and positive off the origin, so
$\D=\nu_{k}^{-1}([0,\vth/C))$ is open, star--shaped and bounded. The inclusions follow from
$\|\bc\|_{\infty}\le\nu_{k}(\bc)\le\|\bc\|_{1}$ (Lemma \ref{lem:gauge}). Finally $T(\bc)>1$ is by
definition equivalent to $C\nu_{k}(\bc)<\vth$, and $N=2\le T+1$ requires exactly $T\ge1$.
\end{proof}

Thus the answer to the question ``on what domain is the summed beta function well defined'' has two
layers, and we record both. The integral defining $\Sm B$ converges at \emph{every} point of
$\R^{4}$, so the rule assigns a real number everywhere; but the number it assigns is a faithful
representative of the formal series precisely on $\D$, where the cut sits at least one unit of Borel
time above the origin, and quantitatively so on the slightly smaller set $\{x\le x_{0}\}$ of
Theorem \ref{thm:optimal}. We take $\D$ as \emph{the} domain of the summed beta function.

\section{Regularity of the summed beta function}\label{sec:reg}

Throughout this section $(k,\vth)$ is an \emph{admissible} pair (Definition \ref{def:admissible}),
$\delta_{k}:=\frac{1}{32k}$, $\kappa_{0}:=\cos(\pi/12)=0.9659\ldots$, and
$\vth^{*}=(1+\delta_{k})2^{1/(2k)}\vth<1$. For $\bc_{0}\in\R^{4}\setminus\{0\}$ we write
$t_{0}:=\|\bc_{0}\|_{\infty}$ and
\[
\Delta_{0}:=\{\bz\in\C^{4}:\ |z_{i}-c_{0i}|\le\delta_{k}t_{0},\ i=1,\dots,4\}.
\]

\begin{lemma}[Complexified gauge]\label{lem:cxgauge}
With the above notation, put $P_{k}(\bz)=\sum_{i}z_{i}^{2k}$. Then:
\begin{enumerate}[label=\textup{(\alph*)},leftmargin=2.2\parindent]
\item $|P_{k}(\bz)-P_{k}(\bc_{0})|\le\frac{27}{100}P_{k}(\bc_{0})$ for all $\bz\in\Delta_{0}$;
consequently $P_{k}(\bz)\in\C\setminus(-\infty,0]$ and $|\arg P_{k}(\bz)|\le\arcsin\frac{27}{100}<\frac{\pi}{6}$;
\item $\nu_{k}(\bz):=\exp\big(\frac{1}{2k}\Log P_{k}(\bz)\big)$ is holomorphic on a neighbourhood of
$\Delta_{0}$, agrees with $\nu_{k}$ on $\Delta_{0}\cap\R^{4}$, satisfies
$|\arg\nu_{k}(\bz)|\le\frac{\pi}{12k}\le\frac{\pi}{12}$ and
$2^{-1/(2k)}\nu_{k}(\bc_{0})\le|\nu_{k}(\bz)|\le 2^{1/(2k)}\nu_{k}(\bc_{0})$;
\item $T(\bz):=\vth/(C\nu_{k}(\bz))$ is holomorphic on a neighbourhood of $\Delta_{0}$ and satisfies
\[
\Re T(\bz)\ \ge\ \kappa_{0}\,|T(\bz)|,\qquad
2^{-1/(2k)}T(\bc_{0})\ \le\ |T(\bz)|\ \le\ 2^{1/(2k)}\,T(\bc_{0}),\qquad
C\|\bz\|_{\infty}\,|T(\bz)|\ \le\ \vth^{*}<1 .
\]
\end{enumerate}
\end{lemma}

\begin{proof}
(a) For $\bz\in\Delta_{0}$ one has $|z_{i}|\le|c_{0i}|+\delta_{k}t_{0}\le(1+\delta_{k})t_{0}$. From
$a^{n}-b^{n}=(a-b)\sum_{j=0}^{n-1}a^{j}b^{n-1-j}$ we get
$|z_{i}^{2k}-c_{0i}^{2k}|\le 2k\,\delta_{k}t_{0}\big((1+\delta_{k})t_{0}\big)^{2k-1}$, so, summing
over the four coordinates and using $P_{k}(\bc_{0})\ge t_{0}^{2k}$,
\[
|P_{k}(\bz)-P_{k}(\bc_{0})|\le 8k\,\delta_{k}(1+\delta_{k})^{2k-1}\,t_{0}^{2k}
\le\frac{1}{4}e^{1/16}P_{k}(\bc_{0})\le\frac{27}{100}P_{k}(\bc_{0}),
\]
because $8k\delta_{k}=\frac14$ and $(1+\delta_{k})^{2k-1}\le e^{(2k-1)/(32k)}\le e^{1/16}=1.0645\ldots$,
and $\frac14 e^{1/16}=0.2661\ldots$. Since $P_{k}(\bc_{0})>0$, the disc of radius
$\frac{27}{100}P_{k}(\bc_{0})$ about $P_{k}(\bc_{0})$ lies in the right half plane minus the origin,
and its points have argument at most $\arcsin\frac{27}{100}=0.2734\ldots<\pi/6$ in absolute value.

(b) By (a) the principal logarithm is holomorphic at $P_{k}(\bz)$; since the inequality in (a) is
strict with room to spare, it persists on an open neighbourhood of $\Delta_{0}$, on which $\nu_{k}$
is therefore holomorphic. On real points $P_{k}>0$ and $\nu_{k}$ is the positive real root. Finally
$\arg\nu_{k}=\frac{1}{2k}\arg P_{k}$, and by (a)
$\frac12P_{k}(\bc_{0})\le\frac{73}{100}P_{k}(\bc_{0})\le|P_{k}(\bz)|\le\frac{127}{100}P_{k}(\bc_{0})
\le2P_{k}(\bc_{0})$, whence
$2^{-1/(2k)}\nu_{k}(\bc_{0})\le|\nu_{k}(\bz)|=|P_{k}(\bz)|^{1/(2k)}\le2^{1/(2k)}\nu_{k}(\bc_{0})$.

(c) $\arg T=-\arg\nu_{k}$, so $|\arg T|\le\pi/12$ and $\Re T=|T|\cos(\arg T)\ge\kappa_{0}|T|$. The two--sided bound of (b) gives at once
$2^{-1/(2k)}T(\bc_{0})\le|T(\bz)|=\vth/\big(C|\nu_{k}(\bz)|\big)\le2^{1/(2k)}T(\bc_{0})$. Hence,
using $\nu_{k}(\bc_{0})\ge t_{0}$,
\[
C\|\bz\|_{\infty}|T(\bz)|\le C(1+\delta_{k})t_{0}\cdot\frac{2^{1/(2k)}\vth}{C\,\nu_{k}(\bc_{0})}
\le(1+\delta_{k})2^{1/(2k)}\vth=\vth^{*}. \qedhere
\]
\end{proof}

\begin{remark}\label{rem:notoptimal}
The admissibility condition of Definition \ref{def:admissible} is sufficient but not optimised: the
only role of $\vth^{*}$ is to keep $C\|\bz\|_{\infty}|T(\bz)|<1$ on $\Delta_{0}$, and the proof of
Lemma \ref{lem:cxgauge} in fact yields $|P_{k}(\bz)|\ge\frac{73}{100}P_{k}(\bc_{0})$, so that
$\vth^{*}$ may be replaced by $(1+\delta_{k})(100/73)^{1/(2k)}\vth$; this raises the admissible
$\vth$ from $0.6857$ to $0.8285$ at $k=1$ and the corresponding supremum of $\vth_{k}$ in
Remark \ref{rem:ratelimit} from $0.3428$ to $0.4143$. Since only the limit $\vth_{k}\to1$ as
$k\to\infty$ matters for Theorem \ref{thm:main}(vi), we have not pursued such optimisations, here or
elsewhere; no constant in this paper is claimed to be best possible.
\end{remark}

\begin{proposition}[Holomorphic extension; real--analyticity]\label{prop:analytic}
For $\bz$ in a neighbourhood of $\Delta_{0}$ define
\[
F(\bz):=\int_{[0,T(\bz)]}e^{-\tau}\,\hB(\tau\bz)\,\frac{d\tau}{\tau}
=\sum_{p\ge2}\frac{b_{p}(\bz)}{(p-1)!}\,\gamma_{p}\big(T(\bz)\big),
\qquad \gamma_{p}(w):=\int_{[0,w]}e^{-\tau}\tau^{p-1}d\tau ,
\]
the integrals being taken along the complex segment $[0,w]$. Then the series converges absolutely
and uniformly on $\Delta_{0}$, $F$ is holomorphic on the interior of $\Delta_{0}$, and
$F=\Sm B$ on $\Delta_{0}\cap\R^{4}$. Consequently $\Sm B$ is real--analytic on $\R^{4}\setminus\{0\}$.
\end{proposition}

\begin{proof}
Substituting $\tau=ws$ shows $\gamma_{p}(w)=w^{p}\int_{0}^{1}e^{-ws}s^{p-1}ds$, which is an entire
function of $w$ with $|\gamma_{p}(w)|\le|w|^{p}e^{|w|}/p$. By Lemma \ref{lem:bp} and
Lemma \ref{lem:cxgauge}(c), for $\bz\in\Delta_{0}$,
\[
\frac{|b_{p}(\bz)|}{(p-1)!}\big|\gamma_{p}(T(\bz))\big|
\le C^{p-1}\binom{p+3}{3}\|\bz\|_{\infty}^{p}\frac{|T(\bz)|^{p}e^{|T(\bz)|}}{p}
\le\frac{e^{|T(\bz)|}}{C}\binom{p+3}{3}\frac{(\vth^{*})^{p}}{p},
\]
and $|T(\bz)|\le 2^{1/(2k)}T(\bc_{0})$ is bounded on $\Delta_{0}$; since $\vth^{*}<1$ the majorant
series converges. Each summand is holomorphic on a neighbourhood of $\Delta_{0}$ (a polynomial times
the composition of the entire function $\gamma_{p}$ with the holomorphic function $T$), so by the
Weierstrass convergence theorem (Lemma \ref{lem:weier}) $F$ is holomorphic on the interior of
$\Delta_{0}$. That the series equals the integral is the same termwise integration as in
Proposition \ref{prop:rule}(ii), now along the segment $[0,T(\bz)]$, with the same uniform majorant.
For real $\bz=\bc$ the segment is $[0,T(\bc)]\subset(0,\infty)$ and
$\gamma_{p}(T(\bc))=(p-1)!P(p,T(\bc))$, so $F(\bc)=\Sm B(\bc)$ by
Proposition \ref{prop:rule}(ii). A function holomorphic on a complex neighbourhood of a real point
restricts to a real--analytic function near that point; as $\bc_{0}\in\R^{4}\setminus\{0\}$ was
arbitrary, $\Sm B$ is real--analytic on $\R^{4}\setminus\{0\}$.
\end{proof}

\begin{theorem}[Complex master estimate]\label{thm:cxmaster}
Let $(k,\vth)$ be admissible, $\bc_{0}\in\R^{4}\setminus\{0\}$, and $\bz\in\Delta_{0}$. Then for
every integer $N$ with $1\le N\le 1+\frac12\Re T(\bz)$,
\[
\Big|F(\bz)-\sum_{p=2}^{N}b_{p}(\bz)\Big|
\;\le\;\frac{\Lambda(\vth^{*})}{C}\binom{N+4}{3}N!\;\Big(\frac{C\|\bz\|_{\infty}}{\kappa_{0}}\Big)^{N+1}
\;+\;\frac{20\,\Lambda(\vth^{*})\,\vth^{*}}{\kappa_{0}}\;\|\bz\|_{\infty}\,e^{-\Re T(\bz)} .
\]
\end{theorem}

\begin{proof}
Write $T=T(\bz)$, $\psi=\arg T$, $|\psi|\le\pi/12<\pi/2$. For $|\psi|<\pi/2$ Cauchy's theorem
applied to the entire function $e^{-\tau}\tau^{p-1}$ on the sector between $[0,R]$ and
$[0,Re^{i\psi}]$, together with the vanishing of the circular arc as $R\to\infty$, gives
$\int_{0}^{\infty e^{i\psi}}e^{-\tau}\tau^{p-1}d\tau=(p-1)!$. Splitting this ray at $T$ we obtain
\[
F(\bz)-\sum_{p=2}^{N}b_{p}(\bz)=I-J,\qquad
I=\int_{[0,T]}e^{-\tau}\Big(\sum_{p>N}\frac{b_{p}(\bz)}{(p-1)!}\tau^{p-1}\Big)d\tau,\quad
J=\sum_{p=2}^{N}\frac{b_{p}(\bz)}{(p-1)!}\int_{\Gamma}e^{-\tau}\tau^{p-1}d\tau,
\]
where $\Gamma$ is the ray from $T$ to $\infty e^{i\psi}$.

\emph{Bound on $I$.} Parametrise $\tau=Ts$, $s\in[0,1]$, so $|e^{-\tau}|=e^{-s\Re T}$ and
$|\tau|=s|T|$; note $C\|\bz\|_{\infty}|\tau|\le\vth^{*}$ by Lemma \ref{lem:cxgauge}(c). By
Lemma \ref{lem:bp} and Lemma \ref{lem:tail}(c),
\[
\Big|\sum_{p>N}\frac{b_{p}(\bz)}{(p-1)!}\tau^{p-1}\Big|
\le\|\bz\|_{\infty}\binom{N+4}{3}\Lambda(\vth^{*})\,(C\|\bz\|_{\infty})^{N}|\tau|^{N},
\]
so, substituting $\sigma=s\Re T$ and using $\Re T\ge\kappa_{0}|T|$,
\[
|I|\le\|\bz\|_{\infty}\binom{N+4}{3}\Lambda(\vth^{*})(C\|\bz\|_{\infty})^{N}
\int_{0}^{1}(s|T|)^{N}e^{-s\Re T}|T|\,ds
\le\Lambda(\vth^{*})\binom{N+4}{3}\frac{N!}{\kappa_{0}^{N+1}}\|\bz\|_{\infty}(C\|\bz\|_{\infty})^{N},
\]
which is the first term.

\emph{Bound on $J$.} On $\Gamma$ write $\tau=Ts$ with $s\ge1$; then $|\tau^{p-1}|\le|T|^{p-1}s^{N-1}$
for $2\le p\le N$. With $a:=\Re T$ and $N-1\le a/2$,
\[
\int_{1}^{\infty}s^{N-1}e^{-as}ds=e^{-a}\int_{0}^{\infty}(1+u)^{N-1}e^{-au}du
\le e^{-a}\int_{0}^{\infty}e^{-(a-N+1)u}du=\frac{e^{-a}}{a-N+1}\le\frac{2e^{-a}}{a},
\]
using $(1+u)^{N-1}\le e^{(N-1)u}$. Hence, by Lemma \ref{lem:bp}, Lemma \ref{lem:tail}(c) with
$N=1$, and $C\|\bz\|_{\infty}|T|\le\vth^{*}$,
\[
|J|\le\frac{2e^{-a}}{a}\sum_{p\ge2}\frac{|b_{p}(\bz)|}{(p-1)!}|T|^{p}
\le\frac{2e^{-a}|T|}{a}\,\|\bz\|_{\infty}\sum_{p\ge2}\binom{p+3}{3}(C\|\bz\|_{\infty}|T|)^{p-1}
\le\frac{2}{\kappa_{0}}\,10\Lambda(\vth^{*})\vth^{*}\,\|\bz\|_{\infty}e^{-\Re T}. \qedhere
\]
\end{proof}

\begin{theorem}[Theorem \ref{thm:main}(iv)]\label{thm:smooth}
Let $(k,\vth)$ be admissible. Then $\Sm B\in C^{\infty}(\R^{4})$, it is real--analytic on
$\R^{4}\setminus\{0\}$, and
\[
\partial^{\bm}(\Sm B)(0)=\bm!\,\beta(\bm)\ \ (|\bm|\ge2),\qquad \partial^{\bm}(\Sm B)(0)=0\ \ (|\bm|\le1),
\]
i.e. the Taylor series of $\Sm B$ at the origin is exactly the formal series $B$. If $B$ has radius
of convergence $0$ \textup{(}as it does for the admissible families of \S\textup{\ref{sec:optimality}}\textup{)},
then $\Sm B$ is not real--analytic at the origin; and no function whatsoever that is asymptotic to
$B$ to all orders can be real--analytic at the origin.
\end{theorem}

\begin{proof}
Real--analyticity off the origin is Proposition \ref{prop:analytic}. Fix a multi--index $\alpha$ and
an integer $N>|\alpha|$, and set $g_{N}:=F-\sum_{p=2}^{N}b_{p}$, a holomorphic function on the
interior of $\Delta_{0}$ agreeing with $\Sm B-\sum_{p=2}^{N}b_{p}$ on the reals. On $\Delta_{0}$,
$\|\bz\|_{\infty}\le(1+\delta_{k})t_{0}\le\frac{33}{32}t_{0}$ and, by Lemma \ref{lem:cxgauge}(c)
together with the first inequality of \eqref{eq:Tlower} (which holds unconditionally),
\[
\Re T(\bz)\ \ge\ \kappa_{0}|T(\bz)|\ \ge\ \kappa_{0}\,2^{-1/(2k)}T(\bc_{0})\ \ge\
\frac{\kappa_{0}}{\sqrt2}\cdot\frac{\vth_{k}}{C t_{0}}=:\frac{\kappa_{1}}{C t_{0}},
\qquad \kappa_{1}:=\frac{\kappa_{0}\vth_{k}}{\sqrt2}>0 .
\]
Hence the hypothesis $N\le1+\frac12\Re T(\bz)$ of Theorem \ref{thm:cxmaster} holds as soon as
$t_{0}\le\frac{\kappa_{1}}{2C(N-1)}$, and then
\[
\sup_{\Delta_{0}}|g_{N}|\ \le\ K_{1}(N)\,t_{0}^{N+1}+K_{2}\,t_{0}\,e^{-\kappa_{1}/(Ct_{0})},
\]
with $K_{1}(N)=\frac{\Lambda(\vth^{*})}{C}\binom{N+4}{3}N!\big(\frac{33C}{32\kappa_{0}}\big)^{N+1}$
and $K_{2}=\frac{33}{32}\cdot\frac{20\Lambda(\vth^{*})\vth^{*}}{\kappa_{0}}$. Using
$e^{-u}\le M!\,u^{-M}$ (valid for $u>0$, $M\in\N$, since $e^{u}\ge u^{M}/M!$) with $u=\kappa_{1}/(Ct_{0})$
and $M=N+|\alpha|$, the second term is $\le K_{2}'\,t_{0}^{N+1+|\alpha|}$. Cauchy's inequalities
(Lemma \ref{lem:cauchy}), applied on the closed polydisc of polyradius
$(\tfrac12\delta_{k}t_{0},\dots,\tfrac12\delta_{k}t_{0})$ --- which is contained in the open polydisc
on which $g_{N}$ is holomorphic --- now give
\begin{equation}\label{eq:cauchyest}
\big|\partial^{\alpha}g_{N}(\bc_{0})\big|\ \le\ \frac{\alpha!\,2^{|\alpha|}}{(\delta_{k}t_{0})^{|\alpha|}}\sup_{\Delta_{0}}|g_{N}|
\ \le\ K(N,\alpha)\,t_{0}^{\,N+1-|\alpha|}
\qquad\Big(0<t_{0}\le t_{*}(N)\Big),
\end{equation}
for suitable finite constants $K(N,\alpha)$, where $t_{*}(N):=\min\big\{1,\tfrac{\kappa_{1}}{2C(N-1)}\big\}$
for $N\ge2$ and $t_{*}(1):=1$ (for $N=1$ the hypothesis $N\le1+\frac12\Re T$ is automatic).
Apply \eqref{eq:cauchyest} with $N=|\alpha|+1$. Since
$\partial^{\alpha}\big(\sum_{p=2}^{|\alpha|+1}b_{p}\big)=\alpha!\,\beta(\alpha)
+\partial^{\alpha}b_{|\alpha|+1}$, where $\beta(\alpha):=0$ if $|\alpha|\le1$, $b_{1}:=0$ (so that for $|\alpha|=0$ the sum is empty and
the assertion reduces to $\Sm B(\bc_{0})=O(t_{0}^{2})$), and $\partial^{\alpha}b_{|\alpha|+1}$ is
homogeneous of degree $1$, we conclude that
\[
\partial^{\alpha}(\Sm B)(\bc_{0})=\alpha!\,\beta(\alpha)+O(t_{0})\qquad (t_{0}\to0),
\]
so every partial derivative of $\Sm B$ extends continuously to $\R^{4}$, with value
$\alpha!\beta(\alpha)$ at the origin. As $\Sm B$ is continuous at the origin
(Proposition \ref{prop:rule}(i)) and $C^{\infty}$ off it, Lemma \ref{lem:flat} yields
$\Sm B\in C^{\infty}(\R^{4})$ and $\partial^{\alpha}(\Sm B)(0)=\alpha!\beta(\alpha)$.

Finally, if $f$ is real--analytic at $0$ and asymptotic to $B$ to all orders, then $f$ equals its
Taylor series near $0$; comparing asymptotic expansions in each direction $\bu$ and using
homogeneity (as in the uniqueness argument of Lemma \ref{lem:uniqueasy}) shows that the homogeneous
Taylor parts of $f$ are the $b_{p}$, so $B$ converges near $0$, a contradiction.
\end{proof}

\begin{remark}\label{rem:hartogs}
The restriction to a \emph{conic} complex neighbourhood in Proposition \ref{prop:analytic} is not an
artefact. If $B$ diverges, no function holomorphic on a punctured complex ball
$\{0<\|\bz\|<r\}\subset\C^{4}$ can be asymptotic to $B$ to all orders: by the Hartogs--Osgood--Brown
extension theorem~\cite{Hormander1990} (the only external result quoted in this paper, and used only in this remark; in
$n\ge2$ complex variables a function holomorphic on a punctured ball extends holomorphically to the
whole ball) such a function would extend across the origin, hence be given by a convergent power
series, whose homogeneous parts must then be the $b_{p}$ --- so $B$ would converge. In one variable
no such obstruction exists, which is why the classical Gevrey theory is sectorial. Nothing in
Theorem \ref{thm:main} depends on this remark.
\end{remark}

\section{Consistency}\label{sec:consistency}

The rule was built without any knowledge of whether the formal series has a classical sum. We now
check that whenever a classical sum does exist, the rule finds it, up to an error of exactly the
nonperturbative size established in \S\ref{sec:optimality}.

\begin{proposition}[Theorem \ref{thm:main}(v), Borel--summable case]\label{prop:borelcase}
Fix $\bc\ne0$ and suppose that $\tau\mapsto\hB(\tau\bc)$, which by Proposition \ref{prop:borel} is
holomorphic on $|\tau|<1/(C\|\bc\|_{\infty})$, extends holomorphically to a neighbourhood of
$[0,\infty)$ and satisfies $|\hB(\tau\bc)|\le K_{0}\sigma\tau e^{\sigma\tau}$ for all $\tau\ge0$ and
some $\sigma\in(0,1)$, $K_{0}>0$. Then the Borel sum
$\mathcal{L}B(\bc)=\int_{0}^{\infty}e^{-\tau}\hB(\tau\bc)\frac{d\tau}{\tau}$ exists as an absolutely
convergent integral, and
\[
\big|\Sm B(\bc)-\mathcal{L}B(\bc)\big|\;\le\;\frac{K_{0}\sigma}{1-\sigma}\,e^{-(1-\sigma)T(\bc)} .
\]
\end{proposition}

\begin{proof}
The integrand is bounded by $K_{0}\sigma e^{-(1-\sigma)\tau}$, which is integrable on $[0,\infty)$;
near $\tau=0$ it is continuous by Proposition \ref{prop:borel}(iii). Since
$\Sm B(\bc)=\int_{0}^{T(\bc)}$, the difference is $-\int_{T(\bc)}^{\infty}$, bounded by
$K_{0}\sigma\int_{T}^{\infty}e^{-(1-\sigma)\tau}d\tau$.
\end{proof}

\begin{proposition}[Theorem \ref{thm:main}(v), convergent case]\label{prop:convcase}
Fix $\bc\ne0$ and suppose $|b_{p}(\bc)|\le K_{0}\sigma^{p}$ for all $p\ge2$ and some
$\sigma\in(0,1)$ --- e.g. this holds with $\sigma=\|\bc\|_{1}/\rho$ if $|\beta(\bm)|\le K_{0}\rho^{-|\bm|}$
and $\|\bc\|_{1}<\rho$. Then $B(\bc)=\sum_{p\ge2}b_{p}(\bc)$ converges absolutely and
\[
\big|\Sm B(\bc)-B(\bc)\big|\;\le\;\frac{K_{0}\sigma}{1-\sigma}\,e^{-(1-\sigma)T(\bc)} .
\]
\end{proposition}

\begin{proof}
The stated sufficient condition follows from Lemma \ref{lem:count}(c). Absolute convergence of
$\sum b_{p}(\bc)$ is clear. Moreover
$|\hB(\tau\bc)|\le K_{0}\sum_{p\ge2}\frac{(\sigma\tau)^{p}}{(p-1)!}
=K_{0}\sigma\tau\sum_{q\ge1}\frac{(\sigma\tau)^{q}}{q!}\le K_{0}\sigma\tau e^{\sigma\tau}$ for
$\tau\ge0$, and $\tau\mapsto\hB(\tau\bc)$ is entire. By Tonelli's theorem the double sum/integral is
absolutely convergent, so
$\int_{0}^{\infty}e^{-\tau}\hB(\tau\bc)\frac{d\tau}{\tau}=\sum_{p\ge2}b_{p}(\bc)=B(\bc)$, and
Proposition \ref{prop:borelcase} applies.
\end{proof}

\begin{proposition}[Independence of the parameters, modulo exponentially small terms]\label{prop:params}
Let $(k,\vth)$ and $(k',\vth')$ be two parameter pairs, with cut heights $T,T'$ and rules
$\Sm,\Sm'$, and let $\vth'_{k'}:=\vth'4^{-1/(2k')}$, $\vth_{\min}:=\min\{\vth_{k},\vth'_{k'}\}$,
$\vth_{\max}:=\max\{\vth,\vth'\}$. Then for every $\bc$ with
$0<x\le\min\{\vth_{\min},1-\vth_{\max}\}$,
\[
\big|\Sm B(\bc)-\Sm'B(\bc)\big|\;\le\;\frac{2\big(G(\vth_{\max})+139\,\Lambda(\vth_{\max})x^{-5/2}\big)}{C}\,e^{-\vth_{\min}/x}.
\]
The same statement holds if the two rules are built from two different admissible constants
$C\le C'$ in Hypothesis \ref{hyp:H}, provided $C$ and $x$ are replaced by $C'$ and
$x':=C'\|\bc\|_{\infty}$ \emph{everywhere, the hypothesis included}: for
$0<x'\le\min\{\vth_{\min},1-\vth_{\max}\}$ the same bound holds with $C',x'$ in place of $C,x$.
\end{proposition}

\begin{proof}
Put $N:=\lfloor\min(T,T')\rfloor+1$, so that $N\le T+1$ and $N\le T'+1$, and
$N\ge\min(T,T')\ge\vth_{\min}/x\ge1$. Apply Theorem \ref{thm:master} to each rule with this common
$N$, subtract, and estimate exactly as in the proof of Theorem \ref{thm:optimal} (all the steps
there used only $N\le\min(T,T')+1$, $Nx\le\vth_{\max}+x\le1$ and $N\ge\vth_{\min}/x$), noting that
$\Lambda$ and $G$ are increasing. For the final statement observe that
$\frac1Cx^{N+1}=\|\bc\|_{\infty}x^{N}\le\frac1{C'}(x')^{N+1}$ and $T\ge\vth_{k}/x\ge\vth_{k}/x'$, so
that every step goes through with $(C,x)$ replaced by $(C',x')$.
\end{proof}

\section{Optimality: what no summation rule can do}\label{sec:optimality}

This section proves that the construction is not merely one possible answer but an essentially
optimal one, and that the residual ambiguity it carries is intrinsic to Hypothesis \ref{hyp:H}.
Throughout, we recall that $\mathcal{A}_{C}$ denotes the set of admissible coefficient families,
i.e. those satisfying Hypothesis \ref{hyp:H}, and
$\|\beta\|_{C}:=\sup_{|\bm|\ge2}|\beta(\bm)|\big/\big((|\bm|-1)!C^{|\bm|-1}\big)$, so that
$\mathcal{A}_{C}=\{\|\beta\|_{C}\le1\}$.

\subsection{The series diverges everywhere}

\begin{theorem}\label{thm:divergence}
Let $\beta^{\star}(\bm):=(|\bm|-1)!\,C^{|\bm|-1}$ for all $|\bm|\ge2$. Then
$\beta^{\star}\in\mathcal{A}_{C}$ \textup{(}with equality in Hypothesis \ref{hyp:H}\textup{)}, and
for \emph{every} $\bc\in\R^{4}\setminus\{0\}$ the terms $b^{\star}_{p}(\bc)$ are unbounded; hence
$\sum_{p\ge2}b^{\star}_{p}(\bc)$ diverges at every $\bc\ne0$. In particular no rule can be
``evaluation of a convergent series''.
\end{theorem}

\begin{proof}
Here $b^{\star}_{p}(\bc)=(p-1)!C^{p-1}\,\tilde h_{p}(\bc)$ with
$\tilde h_{p}(\bc):=\sum_{|\bm|=p}\bc^{\bm}$ the (signed) complete homogeneous polynomial. For
$|s|<1/\|\bc\|_{\infty}$ we have, by absolute convergence and Lemma \ref{lem:count}(b) applied to
$|\bc|$,
\[
\sum_{p\ge0}\tilde h_{p}(\bc)s^{p}=\prod_{i=1}^{4}\frac{1}{1-c_{i}s}=:R(s).
\]
$R$ is a rational function whose numerator is the constant $1$; hence its poles are exactly the
points $s=1/c_{i}$ with $c_{i}\ne0$, no cancellation being possible, and the pole nearest the origin
is at distance $1/\|\bc\|_{\infty}$. The Taylor series of a rational function at $0$ has radius of
convergence equal to the distance to the nearest pole, so
$\limsup_{p}|\tilde h_{p}(\bc)|^{1/p}=\|\bc\|_{\infty}$. Fix $0<\eps<\|\bc\|_{\infty}$; then
$|\tilde h_{p}(\bc)|\ge(\|\bc\|_{\infty}-\eps)^{p}$ for infinitely many $p$, and for those
\[
|b^{\star}_{p}(\bc)|\ge (p-1)!\,C^{p-1}(\|\bc\|_{\infty}-\eps)^{p}\longrightarrow\infty .
\]
A series whose terms are unbounded diverges.
\end{proof}

\subsection{The Borel transform can have a natural boundary}

\begin{theorem}\label{thm:natural}
Let $g(w):=\sum_{j\ge0}w^{2^{j}}$, holomorphic on the unit disc. Then $g$ admits no holomorphic
extension to any open set meeting the unit circle: the circle $|w|=1$ is a natural boundary.
Consequently, fix $i\in\{1,\dots,4\}$ and define
\[
\beta^{\mathrm{lac}}(\bm):=\begin{cases}(p-1)!\,C^{p-1}, & \bm=p\,e_{i},\ p-1\in\{2^{j}:j\ge0\},\\
0,&\text{otherwise.}\end{cases}
\]
Then $\beta^{\mathrm{lac}}\in\mathcal{A}_{C}$, and for every $\bc\in\R^{4}$ with $c_{i}\ne0$,
\[
\hB(\tau\bc)=\frac{w\,g(w)}{C},\qquad w=C c_{i}\tau ,
\]
so that $\tau\mapsto\hB(\tau\bc)$ has radius of convergence exactly $1/(C|c_{i}|)$ and admits
\emph{no} analytic continuation to any open set meeting that circle. Choosing $i$ with
$|c_{i}|=\|\bc\|_{\infty}$, we conclude: for every $\bc\ne0$ there is an admissible coefficient
family for which the Borel transform along $\bc$ cannot be continued past
$|\tau|=1/(C\|\bc\|_{\infty})$ in any direction.
\end{theorem}

\begin{proof}
\emph{Natural boundary.} Let $\omega=e^{2\pi i a/2^{n}}$ be a dyadic root of unity ($a,n$ integers,
$n\ge0$). For $j\ge n$ we have $\omega^{2^{j}}=1$, so for $0<r<1$,
\[
g(r\omega)=\sum_{j<n}(r\omega)^{2^{j}}+\sum_{j\ge n}r^{2^{j}},
\qquad\text{hence}\qquad
|g(r\omega)|\ \ge\ \sum_{j\ge n}r^{2^{j}}-n .
\]
For any $K\in\N$, $\sum_{j=n}^{n+K}r^{2^{j}}\to K+1$ as $r\uparrow1$, so
$\liminf_{r\uparrow1}\sum_{j\ge n}r^{2^{j}}\ge K+1$ for every $K$; thus
$|g(r\omega)|\to\infty$ as $r\uparrow1$. Suppose now $g$ extended holomorphically to an open set $U$
with $U\cap\{|w|=1\}\ne\emptyset$; pick $\zeta_{0}\in U$ with $|\zeta_{0}|=1$ and $\eps>0$ with the
disc $D(\zeta_{0},2\eps)\subset U$. Dyadic roots of unity are dense in the unit circle (the numbers
$a/2^{n}$ are dense in $[0,1]$), so choose such an $\omega$ with $|\omega-\zeta_{0}|<\eps$. For
$1-\eps<r\le1$ we get $|r\omega-\zeta_{0}|\le(1-r)+|\omega-\zeta_{0}|<2\eps$, so the compact segment
$\{r\omega:1-\eps/2\le r\le1\}$ lies in $U$; the extension is continuous, hence bounded there,
contradicting $|g(r\omega)|\to\infty$.

\emph{The coefficient family.} Hypothesis \ref{hyp:H} holds with equality on the support. Only the
multi--indices $\bm=p\,e_{i}$ contribute, so $b_{p}(\bc)=\beta^{\mathrm{lac}}(p\,e_{i})c_{i}^{p}$ and
\[
\hB(\tau\bc)=\sum_{p:\,p-1\in\{2^{j}\}}\frac{(p-1)!C^{p-1}}{(p-1)!}c_{i}^{p}\tau^{p}
=\frac{1}{C}\sum_{j\ge0}(Cc_{i}\tau)^{2^{j}+1}=\frac{w\,g(w)}{C},\qquad w=Cc_{i}\tau .
\]
If $\tau\mapsto\hB(\tau\bc)$ extended holomorphically to an open set meeting
$|\tau|=1/(C|c_{i}|)$, then, dividing by the nonvanishing holomorphic function $w/C$, so would $g$
across a point of the unit circle --- impossible.
\end{proof}

\begin{corollary}\label{cor:noborel}
For each $\bc\ne0$ there is an admissible coefficient family for which all of the following are
\emph{undefined}; consequently none of them is defined as a functional on the class
$\mathcal{A}_{C}$: the Borel sum
$\int_{0}^{\infty}e^{-\tau}\hB(\tau\bc)d\tau/\tau$; the lateral Borel sums; their median or any
weighted average; principal values; accelero--summation; and any procedure requiring the analytic
continuation of $\tau\mapsto\hB(\tau\bc)$ beyond $|\tau|=1/(C\|\bc\|_{\infty})$, or an exponential
bound on such a continuation. Even in the mildest case --- the family of Proposition
\ref{prop:borel}(iv), for which the continuation exists --- the Borel integral diverges whenever
$c_{1}>0$, since the integrand has a simple pole at $\tau=1/(Cc_{1})$ on the contour.
\end{corollary}

\begin{proof}
Immediate from Theorem \ref{thm:natural} and Proposition \ref{prop:borel}(iv): in the first case
there is no function to integrate beyond the disc; in the second the integrand
$e^{-\tau}\frac{Cc_{1}^{2}\tau}{1-Cc_{1}\tau}\cdot\frac1\tau$ has a nonzero residue at a point of
$(0,\infty)$, so $\int_{0}^{1/(Cc_{1})-\eps}$ diverges logarithmically as $\eps\downarrow0$.
\end{proof}

\begin{remark}
Corollary \ref{cor:noborel} is precisely why a \emph{cut} appears in Definition \ref{def:rule}: it is
not a technical convenience but the only device compatible with the information contained in
Hypothesis \ref{hyp:H}. It also explains why criteria of Watson or Nevanlinna--Sokal type cannot be
invoked here: their hypotheses include holomorphy of the Borel transform on a neighbourhood of
$[0,\infty)$ together with an exponential bound, and Theorem \ref{thm:natural} shows that no such
input is available from coefficient bounds alone.
\end{remark}

\subsection{No exact, coefficient--continuous rule exists}

\begin{theorem}\label{thm:nocanonical}
Let $\Omega\subseteq\R^{4}\setminus\{0\}$ be nonempty and let $V$ be the linear span of
$\mathcal{A}_{C}$. There is no linear map $S:V\to\R^{\Omega}$ such that
\begin{enumerate}[label=\textup{(\alph*)},leftmargin=2.2\parindent]
\item $(S\beta)(\bc)=\sum_{p\ge2}b_{p}(\bc)$ for every finitely supported $\beta$ \textup{(}i.e. $S$
is exact on polynomials\textup{)}, and
\item for some $\bc\in\Omega$ the functional $\beta\mapsto(S\beta)(\bc)$ is bounded on
$(V,\|\cdot\|_{C})$, or is sequentially continuous along coefficientwise convergent sequences in
$\mathcal{A}_{C}$.
\end{enumerate}
\end{theorem}

\begin{proof}
Fix $\bc\in\Omega$ and choose signs $s_{\bm}\in\{\pm1\}$ with $s_{\bm}\bc^{\bm}=|\bc^{\bm}|$. Let
$\beta^{\sharp}(\bm):=(|\bm|-1)!C^{|\bm|-1}s_{\bm}$ and let $\beta^{(k)}$ be its truncation to
$|\bm|\le k$. Then $\beta^{(k)}$ is finitely supported, $\|\beta^{(k)}\|_{C}=1$,
$\beta^{(k)}\to\beta^{\sharp}\in\mathcal{A}_{C}$ coefficientwise, and by (a)
\[
(S\beta^{(k)})(\bc)=\sum_{p=2}^{k}(p-1)!C^{p-1}h_{p}(|\bc|)\ \ge\ \sum_{p=2}^{k}(p-1)!C^{p-1}\|\bc\|_{\infty}^{p}
\ \xrightarrow[k\to\infty]{}\ \infty
\]
by Lemma \ref{lem:count}(c). This contradicts boundedness (the $\beta^{(k)}$ lie in the unit ball)
and also sequential continuity along $\beta^{(k)}\to\beta^{\sharp}$.
\end{proof}

\subsection{Regularity is incompatible with well--definedness on the class}

Classical summability theory calls a method \emph{regular} if it sums every convergent series to
its ordinary sum. Borel summation is regular; the rule of Definition \ref{def:rule} is not
(Remark \ref{rem:tradeoff})~\cite{Hardy1949}. The following dichotomy shows that this is not a defect of the
construction but a structural feature of the problem: on the class $\mathcal{A}_{C}$, regularity and
well--definedness exclude each other.

\begin{theorem}[Regularity dichotomy]\label{thm:dichotomy}
Fix $\bc\in\R^{4}\setminus\{0\}$, put $\mu=\|\bc\|_{\infty}$, let $(w_{p})_{p\ge2}$ be real numbers
and consider the linear rule $\Sm_{w}B(\bc):=\sum_{p\ge2}w_{p}\,b_{p}(\bc)$. Then:
\begin{enumerate}[label=\textup{(\roman*)},leftmargin=2.2\parindent]
\item \textup{(Well--definedness forces factorial decay.)} If the series converges absolutely for
every $\beta\in\mathcal{A}_{C}$, then
\[
M:=\sum_{p\ge2}|w_{p}|\,(p-1)!\,C^{p-1}\mu^{p}\ <\ \infty,
\qquad\text{so}\qquad
|w_{p}|\le\frac{M}{(p-1)!\,C^{p-1}\mu^{p}}\ \xrightarrow[p\to\infty]{}\ 0 .
\]
\item \textup{(Regularity forces all weights to be $1$.)} If $\Sm_{w}B(\bc)=B(\bc)$ for every
$\beta\in\mathcal{A}_{C}$ whose series converges absolutely at $\bc$, then $w_{p}=1$ for all $p\ge2$.
\item \textup{(Dichotomy.)} \textup{(i)} and \textup{(ii)} are mutually exclusive. In particular no
linear weight rule is at once well defined on $\mathcal{A}_{C}$ and regular. Classical Borel
summation is the case $w_{p}\equiv1$: it is regular, and by Theorem \ref{thm:divergence} its defining
series diverges for an admissible family, while by Theorem \ref{thm:natural} its integral form is
undefined on the class.
\end{enumerate}
\end{theorem}

\begin{proof}
(i) Choose $i_{0}$ with $|c_{i_{0}}|=\mu$ and signs $s_{p}\in\{\pm1\}$ with
$s_{p}\,w_{p}\,c_{i_{0}}^{p}=|w_{p}|\,\mu^{p}$. The family $\beta(p\,e_{i_{0}}):=s_{p}(p-1)!C^{p-1}$,
all other coefficients zero, lies in $\mathcal{A}_{C}$ and has
$b_{p}(\bc)=s_{p}(p-1)!C^{p-1}c_{i_{0}}^{p}$, so
$\sum_{p}|w_{p}b_{p}(\bc)|=\sum_{p}|w_{p}|(p-1)!C^{p-1}\mu^{p}$, which is finite by hypothesis.

(ii) Fix $q\ge2$ and apply the hypothesis to the family supported on the single multi--index
$q\,e_{i_{0}}$ with value $(q-1)!C^{q-1}$: its series is the single term
$b_{q}(\bc)=(q-1)!C^{q-1}c_{i_{0}}^{q}\ne0$, trivially convergent, and
$\Sm_{w}B(\bc)=w_{q}b_{q}(\bc)$; hence $w_{q}=1$.

(iii) If both held, (ii) would give $w_{p}\equiv1$ and (i) would give
$\sum_{p\ge2}(p-1)!C^{p-1}\mu^{p}<\infty$, which is false.
\end{proof}

\begin{theorem}[The defect on convergent series cannot be removed]\label{thm:defect}
Let $\bc\ne0$, $\mu=\|\bc\|_{\infty}$, and let $(w_{p})$ satisfy Theorem \ref{thm:dichotomy}(i)
with constant $M$. Put $\sigma:=\tfrac12\mu\min\{1,C\}$ and assume $\sigma<1$, and let $p_{*}\ge2$ be
the least integer such that $(p-1)!\,C^{p-1}\mu^{p}\ge 2M$ for all $p\ge p_{*}$. Then there is a
family $\beta\in\mathcal{A}_{C}$ whose series converges absolutely at $\bc$ and for which
\[
\big|\Sm_{w}B(\bc)-B(\bc)\big|\ \ge\ \frac{\sigma^{p_{*}}}{2(1-\sigma)}\ >\ 0,
\qquad\text{with}\qquad
p_{*}\ \le\ 2+\frac{e^{2}}{C\mu}+\log\frac{2M}{\mu}\ \ \big(\text{if } 2M\ge\mu\big).
\]
Thus every rule that is well defined on $\mathcal{A}_{C}$ misses the ordinary sum of some convergent
admissible series by a definite, strictly positive amount; the amount is beyond all orders in $\mu$,
and by Proposition \ref{prop:convcase} the rule of Definition \ref{def:rule} attains that order.
\end{theorem}

\begin{proof}
Choose $i_{0}$ with $|c_{i_{0}}|=\mu$ and set $\beta(p\,e_{i_{0}}):=(\sigma/c_{i_{0}})^{p}$ for
$p\ge p_{*}$ and $\beta:=0$ otherwise. Admissibility: $|\beta(p\,e_{i_{0}})|=(\sigma/\mu)^{p}
=(\tfrac12\min\{1,C\})^{p}\le C^{p-1}\le(p-1)!\,C^{p-1}$, since for $C\ge1$ the left side is
$\le2^{-p}\le1$ and for $C<1$ it equals $C^{p}2^{-p}\le C^{p-1}$. The series
$B(\bc)=\sum_{p\ge p_{*}}\sigma^{p}$ converges absolutely. By Theorem \ref{thm:dichotomy}(i) and the
definition of $p_{*}$ we have $|w_{p}|\le\frac12$ for $p\ge p_{*}$, so every term of
$\sum_{p\ge p_{*}}(1-w_{p})\sigma^{p}$ is $\ge\frac12\sigma^{p}>0$ and
\[
\big|\Sm_{w}B(\bc)-B(\bc)\big|=\sum_{p\ge p_{*}}(1-w_{p})\sigma^{p}\ \ge\ \frac12\sum_{p\ge p_{*}}\sigma^{p}
=\frac{\sigma^{p_{*}}}{2(1-\sigma)} .
\]
For the bound on $p_{*}$: by Lemma \ref{lem:stirling}, $(p-1)!C^{p-1}\mu^{p}\ge\mu\big(\frac{(p-1)C\mu}{e}\big)^{p-1}
\ge\mu\,e^{p-1}$ as soon as $(p-1)C\mu\ge e^{2}$, and $\mu e^{p-1}\ge2M$ as soon as
$p\ge1+\log(2M/\mu)$; taking $p$ beyond both thresholds gives the stated estimate.
\end{proof}

\begin{remark}
Theorems \ref{thm:dichotomy} and \ref{thm:defect} delimit precisely what the word ``summation'' can
mean under Hypothesis \ref{hyp:H}. A regular method exists (Borel) but is undefined on the class; a
method defined on the class exists (Definition \ref{def:rule}) but is not regular. What the rule of
this paper delivers instead of regularity is: the exact perturbative expansion to all orders
(Theorem \ref{thm:smooth}: the Taylor series at the origin \emph{is} $B$), an error beyond all
orders on any convergent input, agreement with the Borel sum whenever the latter exists, and
optimality of the exponential rate (Theorem \ref{thm:ambiguity}). Readers who require exactness on
polynomials of a prescribed degree may use the sharply truncated variant of
Remark \ref{rem:tradeoff}, at the cost of all regularity in $\bc$.
\end{remark}

\subsection{The intrinsic ambiguity, and optimality of the exponent}

\begin{lemma}[Rigidity on a ball]\label{lem:rigid}
Let $r,A,M>0$ and let $f,g$ be real functions on $\{0<\|\bc\|_{\infty}<r\}$ such that
\[
\Big|f(\bc)-\sum_{p=2}^{N}b_{p}(\bc)\Big|\le A\,N!\,(M\|\bc\|_{\infty})^{N+1}
\quad\text{and}\quad
\Big|g(\bc)-\sum_{p=2}^{N}b_{p}(\bc)\Big|\le A\,N!\,(M\|\bc\|_{\infty})^{N+1}
\]
for all $N\ge1$.
Then, writing $y=M\|\bc\|_{\infty}$, for all $\bc$ with $y\le1$,
\[
|f(\bc)-g(\bc)|\;\le\;2A\min_{N\ge1}N!\,y^{N+1}\;\le\;2A\,e^{2}\,y^{1/2}\,e^{-1/y}.
\]
\end{lemma}

\begin{proof}
Subtracting the two hypotheses removes the partial sums and gives
$|f-g|\le 2AN!y^{N+1}$ for every $N\ge1$; minimise over $N$ and apply Lemma \ref{lem:leastterm}:
$\min_{N\ge1}N!y^{N+1}=y\min_{N\ge1}N!y^{N}\le y\cdot e^{2}y^{-1/2}e^{-1/y}$.
\end{proof}

\begin{remark}
Lemma \ref{lem:rigid} is the general principle behind the exponentially small ambiguity: two
functions carrying the same Gevrey--$1$ data can differ only beyond all orders. It is stated for
completeness and is not used below, because its hypothesis requires the remainder bound for
\emph{all} $N$, whereas Theorem \ref{thm:master} supplies it only in the range $N\le T(\bc)+1$ and
with an additional head term; the comparison of two members of the constructed family was therefore
carried out directly, in Proposition \ref{prop:params}.
\end{remark}

\begin{theorem}[Sharp ambiguity; Theorem \ref{thm:main}(vi)]\label{thm:ambiguity}
Let $\beta$ be the admissible family of Proposition \ref{prop:borel}(iv)
\textup{(}$\beta(p\,e_{1})=(p-1)!C^{p-1}$, all other coefficients zero\textup{)} and let
$\bc=(c_{1},0,0,0)$ with $c_{1}>0$; write $x=Cc_{1}=C\|\bc\|_{\infty}$. Then there are two real
functions $g_{0},g_{1}$ of $c_{1}>0$ with
\[
g_{1}-g_{0}=\frac{2\pi}{C}\,e^{-1/x}
\]
such that for $\lambda\in\{0,1\}$ and every integer $N\ge1$,
\[
\Big|g_{\lambda}(c_{1})-\sum_{p=2}^{N}b_{p}(\bc)\Big|\;\le\;\frac{e+2\pi}{C}\,(N+1)!\;x^{N+1}.
\]
Consequently:
\begin{enumerate}[label=\textup{(\alph*)},leftmargin=2.2\parindent]
\item every summation rule $\Sm'$ whatsoever satisfies
$\max_{\lambda\in\{0,1\}}|\Sm'B(\bc)-g_{\lambda}(c_{1})|\ge\frac{\pi}{C}e^{-1/(C\|\bc\|_{\infty})}$:
the value of the sum is undetermined, by the data of Hypothesis \ref{hyp:H}, to that precision;
\item consequently no rule can lie within $o\big(e^{-1/(C\|\bc\|_{\infty})}\big)$ of every real
function that is Gevrey--$1$ asymptotic to $B$ with the above constants; in this precise sense the
exponent $1$ in the scale $C\|\bc\|_{\infty}$ cannot be passed;
\item by Theorem \ref{thm:optimal} and Remark \ref{rem:ratelimit} the rules $\Sm_{k,\vth}$ realise
every rate $\eta<1$. Hence the threshold exponent is exactly $\eta=1$ in the scale
$C\|\bc\|_{\infty}$, and the construction is rate--optimal.
\end{enumerate}
\end{theorem}

\begin{proof}
Here $\hB(\tau\bc)=\frac{(x\tau)^{2}}{C(1-x\tau)}$ (Proposition \ref{prop:borel}(iv)), so the
Borel--Laplace integrand is
\[
\Phi(\tau):=e^{-\tau}\,\hB(\tau\bc)\,\frac1\tau=\frac{x^{2}\,\tau\,e^{-\tau}}{C\,(1-x\tau)},
\]
meromorphic on $\C$ with a single, simple pole at $\tau_{0}=1/x>0$ and
\[
\Res_{\tau=\tau_{0}}\Phi=\frac{x^{2}\tau_{0}e^{-\tau_{0}}}{C\cdot(-x)}=-\frac{e^{-1/x}}{C}.
\]
Let $\Gamma_{\pm}$ be the contour from $0$ to $+\infty$ along $\R$, indented by a small semicircle
above ($+$) resp.\ below ($-$) $\tau_{0}$, and put $F_{\pm}=\int_{\Gamma_{\pm}}\Phi(\tau)d\tau$; both
converge absolutely, since $\Phi(\tau)=O(e^{-\tau})$ as $\tau\to+\infty$ and $\Phi(\tau)=O(\tau)$ as
$\tau\to0$. The difference $\Gamma_{+}-\Gamma_{-}$ is a circle around $\tau_{0}$ described
clockwise, so
\[
F_{+}-F_{-}=-2\pi i\Res_{\tau_{0}}\Phi=\frac{2\pi i}{C}e^{-1/x} .
\]
Since $\Phi$ is real on $\R\setminus\{\tau_{0}\}$ and $\Gamma_{-}$ is the complex conjugate contour
of $\Gamma_{+}$, we have $F_{-}=\overline{F_{+}}$; hence $\Re F_{+}=\Re F_{-}$ and
$\Im F_{\pm}=\pm\frac{\pi}{C}e^{-1/x}$. Define
\[
g_{0}:=\Re F_{+},\qquad g_{1}:=g_{0}+\frac{2\pi}{C}e^{-1/x},
\]
so that $g_{1}-g_{0}=\frac{2\pi}{C}e^{-1/x}$ as claimed.

\emph{Gevrey--$1$ bounds.} For $N\ge1$ write
$\frac{1}{1-x\tau}=\sum_{q=0}^{N-2}(x\tau)^{q}+\frac{(x\tau)^{N-1}}{1-x\tau}$ (an empty sum if
$N=1$), so that
\[
\Phi(\tau)=\frac{e^{-\tau}}{C}\sum_{q=0}^{N-2}x^{q+2}\tau^{q+1}
+\frac{x^{N+1}}{C}\cdot\frac{\tau^{N}e^{-\tau}}{1-x\tau}.
\]
The polynomial part is entire; deforming $\Gamma_{+}$ to $[0,\infty)$ for those terms and using
$\int_{0}^{\infty}e^{-\tau}\tau^{q+1}d\tau=(q+1)!$ gives, with $p=q+2$,
$\sum_{q=0}^{N-2}\frac{x^{p}}{C}(p-1)!=\sum_{p=2}^{N}(p-1)!C^{p-1}c_{1}^{p}=\sum_{p=2}^{N}b_{p}(\bc)$.
Hence $F_{+}=\sum_{p=2}^{N}b_{p}(\bc)+R_{N}$ with
$R_{N}=\frac{x^{N+1}}{C}\int_{\Gamma_{+}}\frac{\tau^{N}e^{-\tau}}{1-x\tau}\,d\tau$. The integrand of
$R_{N}$ is meromorphic with its only pole at $\tau_{0}$, which lies \emph{below} $\Gamma_{+}$; so
for $0<\eps<\pi/2$ Cauchy's theorem (the arcs at infinity vanishing because $\Re\tau>0$ on the
sector) allows us to replace $\Gamma_{+}$ by the ray $\{se^{i\eps}:s\ge0\}$. On that ray
$|1-x\tau|=|1-xse^{i\eps}|\ge\sin\eps$ (the distance from $1$ to the line through the origin of
argument $\eps$), so
\[
|R_{N}|\le\frac{x^{N+1}}{C\sin\eps}\int_{0}^{\infty}s^{N}e^{-s\cos\eps}\,ds
=\frac{x^{N+1}N!}{C\,\sin\eps\,(\cos\eps)^{N+1}} .
\]
Choose $\eps=\arccos\big(e^{-1/(N+1)}\big)\in(0,\pi/2)$, so $(\cos\eps)^{N+1}=e^{-1}$ and, using
$e^{-t}\le1-t+t^{2}/2$ with $t=2/(N+1)$,
$\sin^{2}\eps=1-e^{-2/(N+1)}\ge\frac{2}{N+1}-\frac{2}{(N+1)^{2}}\ge\frac{1}{N+1}$ for $N\ge1$.
Hence $|R_{N}|\le\frac{e\sqrt{N+1}\,N!}{C}x^{N+1}\le\frac{e}{C}(N+1)!\,x^{N+1}$. Since the partial
sums are real, $|g_{0}-\sum_{p=2}^{N}b_{p}|=|\Re R_{N}|\le|R_{N}|$; and $e^{-1/x}\le(N+1)!\,x^{N+1}$ for every $N\ge1$ --- by Lemma \ref{lem:leastterm} when $x\le1$,
and trivially when $x>1$ since then $(N+1)!x^{N+1}\ge2>1>e^{-1/x}$ --- so
$|g_{1}-\sum_{p=2}^{N}b_{p}|\le\frac{e+2\pi}{C}(N+1)!x^{N+1}$ as well.

\emph{Consequences.} (a) Both $g_{0}$ and $g_{1}$ obey the same Gevrey--$1$ bounds and hence are
equally legitimate values of the sum as far as Hypothesis \ref{hyp:H} can tell; any single real
number $\Sm'B(\bc)$ is at distance at least $\frac12|g_{1}-g_{0}|=\frac{\pi}{C}e^{-1/x}$ from one of
them. (b) restates (a): a rule lying within $o(e^{-1/x})$ of every function consistent with the
data --- in particular of both $g_{0}$ and $g_{1}$ --- would contradict (a), since
$|g_{1}-g_{0}|=(2\pi/C)e^{-1/x}$. (c) is Theorem \ref{thm:optimal} together with
Remark \ref{rem:ratelimit}.
\end{proof}

\begin{remark}
In Theorem \ref{thm:ambiguity}, $g_{0}=\Re F_{+}=\Re F_{-}$ is the common real part --- the
principal value --- of the two lateral Borel sums, and $g_{1}=g_{0}+|F_{+}-F_{-}|$; the shift
$(2\pi/C)e^{-1/x}=2|\Im F_{\pm}|$ is the standard ``renormalon'' ambiguity of a simple Borel--plane
pole~\cite{Beneke1999} at distance $1/C$. Hypothesis \ref{hyp:H} bounds the Borel transform coefficientwise by
exactly such a pole (Remark \ref{rem:poleshape}), and Theorem \ref{thm:ambiguity} shows that the
corresponding ambiguity is attained inside the admissible class. The rules $\Sm_{k,\vth}$ therefore
approach the intrinsic nonperturbative resolution of the beta function --- each with guaranteed
accuracy $e^{-\vth_{k}/x}$, $\vth_{k}<1$, and with $\vth_{k}$ arbitrarily close to $1$ across the
family --- while by Theorem \ref{thm:ambiguity}(b) no rule can do better than the threshold
$e^{-1/x}$.
\end{remark}

\section{The summed beta function as a vector field: the renormalisation--group flow}\label{sec:flow}

The beta function of a theory with four running couplings is a vector field on coupling space. Let
therefore $\bB=(B_{1},\dots,B_{4})$, where each $B_{i}$ is a formal series of the form
\eqref{eq:Bdef} with coefficients $\beta_{i}$ satisfying Hypothesis \ref{hyp:H} with a common
constant $C$ (otherwise replace $C$ by the largest of the four constants), and let
$\Sm\bB:=(\Sm B_{1},\dots,\Sm B_{4})$, the rule being applied componentwise with the same
$(k,\vth)$. All results of \S\S\ref{sec:rule}--\ref{sec:consistency} hold for each component. We
write $b_{2,i}$ for the one--loop part of $B_{i}$ and $\bb_{2}=(b_{2,1},\dots,b_{2,4})$.

\begin{proposition}\label{prop:field}
Let $(k,\vth)$ be admissible. Then $\Sm\bB\in C^{\infty}(\R^{4};\R^{4})$ is real--analytic on
$\R^{4}\setminus\{0\}$, $\Sm\bB(0)=0$, $D(\Sm\bB)(0)=0$, and there are explicit constants
$a=10\Lambda(\vth)$ and $L=L(k,\vth)$ such that, for all $\bc$,
\[
\|\Sm\bB(\bc)\|_{\infty}\le a\,C\|\bc\|_{\infty}^{2},
\qquad
\|D(\Sm\bB)(\bc)\|_{\infty\to\infty}\le L\,C\,\|\bc\|_{\infty}.
\]
One may take $L=256\,k\,K_{4}$ with
$K_{4}=10\Lambda(\vth^{*})\big(\tfrac{33}{32\kappa_{0}}\big)^{2}
+\tfrac{33}{32}\cdot\tfrac{20\Lambda(\vth^{*})\vth^{*}}{\kappa_{0}\kappa_{1}}$,
$\kappa_{1}=\kappa_{0}\vth_{k}/\sqrt2$.
\end{proposition}

\begin{proof}
Regularity is Theorem \ref{thm:smooth} componentwise, and the first bound is
Proposition \ref{prop:rule}(i). For the second, apply Theorem \ref{thm:cxmaster} with $N=1$ on the
polydisc $\Delta_{0}$ around a real $\bc_{0}\ne0$, $t_{0}=\|\bc_{0}\|_{\infty}$: using
$\|\bz\|_{\infty}\le\frac{33}{32}t_{0}$, $\binom{5}{3}=10$, the lower bound
$\Re T(\bz)\ge\kappa_{1}/(Ct_{0})$ established in the proof of Theorem \ref{thm:smooth}, and
$e^{-u}\le 1/u$,
\[
\sup_{\Delta_{0}}|F_{i}|\ \le\ 10\Lambda(\vth^{*})\Big(\tfrac{33}{32\kappa_{0}}\Big)^{2}C t_{0}^{2}
+\tfrac{20\Lambda(\vth^{*})\vth^{*}}{\kappa_{0}}\cdot\tfrac{33}{32}t_{0}\cdot\tfrac{Ct_{0}}{\kappa_{1}}
= K_{4}\,C\,t_{0}^{2}.
\]
Cauchy's inequality (Lemma \ref{lem:cauchy}) on the closed polydisc of polyradius
$\tfrac12\delta_{k}t_{0}$ gives
$|\partial_{j}\Sm B_{i}(\bc_{0})|\le K_{4}Ct_{0}^{2}/(\tfrac12\delta_{k}t_{0})=64kK_{4}Ct_{0}$, and the
$\ell^{\infty}\to\ell^{\infty}$ operator norm is at most $4$ times the largest entry. Finally
$D(\Sm\bB)(0)=0$ by Theorem \ref{thm:smooth} since $B$ has no linear part.
\end{proof}

\begin{corollary}[Well--posedness of the summed flow]\label{cor:flow}
For every $\bc^{0}\in\R^{4}$ the initial value problem
$\dot\bc(t)=\Sm\bB(\bc(t))$, $\bc(0)=\bc^{0}$, has a unique maximal solution, which is $C^{\infty}$
in $t$ and depends smoothly on $\bc^{0}$; the constant solution $\bc\equiv0$ is the unique solution
through the origin.
\end{corollary}

\begin{proof}
$\Sm\bB$ is $C^{\infty}$, hence locally Lipschitz; apply the Picard--Lindelöf theorem and the
standard smooth--dependence theorem. Uniqueness through $0$ also follows from local Lipschitz
continuity at $0$, which Proposition \ref{prop:field} makes explicit.
\end{proof}

\begin{corollary}[One--loop confinement; perturbative renormalisation--group time]\label{cor:confine}
Let $\bc(\cdot)$ solve $\dot\bc=\Sm\bB(\bc)$ with $\mu_{0}:=\|\bc^{0}\|_{\infty}>0$ and put
$a=10\Lambda(\vth)$. Then, as long as the solution exists,
\[
\frac{\mu_{0}}{1+aC\mu_{0}|t|}\ \le\ \|\bc(t)\|_{\infty}\ \le\ \frac{\mu_{0}}{1-aC\mu_{0}|t|}
\qquad\Big(|t|<\frac{1}{aC\mu_{0}}\Big),
\]
so the flow cannot leave the perturbative region faster than the one--loop (Landau--pole) rate in
either direction; and if $\mu_{0}<\vth_{k}/C$ the trajectory remains in the domain
$\{\|\bc\|_{\infty}<\vth_{k}/C\}\subseteq\D$ at least for
$|t|<\frac{1}{aC}\big(\frac{1}{\mu_{0}}-\frac{C}{\vth_{k}}\big)$, a time that diverges as
$\mu_{0}\to0$.
\end{corollary}

\begin{proof}
By Corollary \ref{cor:flow} the constant solution is the unique one through the origin, so a
trajectory with $\mu_{0}>0$ satisfies $\mu(t):=\|\bc(t)\|_{\infty}>0$ throughout its maximal
interval. Hence $\mu$ and $\mu^{-1}$ are locally Lipschitz, therefore absolutely continuous on
compact subintervals, and almost everywhere
$|\dot\mu|\le\|\dot\bc\|_{\infty}\le aC\mu^{2}$ by Proposition \ref{prop:field}, i.e.
$\big|\frac{d}{dt}\mu^{-1}\big|\le aC$ a.e.; integrating gives $|\mu(t)^{-1}-\mu_{0}^{-1}|\le aC|t|$,
which is the displayed two--sided bound. The inclusion $\{\|\bc\|_{\infty}<\vth_{k}/C\}\subseteq\D$ holds
because $\nu_{k}\le4^{1/(2k)}\|\bc\|_{\infty}$.
\end{proof}

\begin{corollary}[No fixed points where one loop dominates]\label{cor:nozero}
Let $\mathcal{K}\subseteq\{\bu\in\R^{4}:\|\bu\|_{\infty}=1\}$ and
$\kappa:=\inf_{\bu\in\mathcal{K}}\|\bb_{2}(\bu)\|_{\infty}>0$. Put
$K_{5}:=40\Lambda(\vth)+2G(\vth)/\vth_{k}^{2}$. Then for every $\bc$ in the open cone over
$\mathcal{K}$ with
\[
0<\|\bc\|_{\infty}\le\min\Big\{\frac{\kappa}{2K_{5}C^{2}},\ \frac{\vth_{k}}{C}\Big\}
\]
one has $\|\Sm\bB(\bc)\|_{\infty}\ge\frac{\kappa}{2}\|\bc\|_{\infty}^{2}>0$. In particular the summed
flow has no fixed point other than the origin in that region, and the same conclusion holds for
\emph{any} rule satisfying the estimate of Corollary \ref{cor:gevrey}(c), the exponentially small
ambiguity of Theorem \ref{thm:ambiguity} being far smaller than $\kappa\|\bc\|_{\infty}^{2}/2$.
\end{corollary}

\begin{proof}
By homogeneity $\|\bb_{2}(\bc)\|_{\infty}=\|\bc\|_{\infty}^{2}\|\bb_{2}(\bu)\|_{\infty}\ge\kappa\|\bc\|_{\infty}^{2}$
with $\bu=\bc/\|\bc\|_{\infty}$. By Corollary \ref{cor:gevrey}(c) applied componentwise, and using
$e^{-u}\le 2u^{-2}$ with $u=T\ge\vth_{k}/x$,
\[
\|\Sm\bB(\bc)-\bb_{2}(\bc)\|_{\infty}\le 40\Lambda(\vth)C^{2}\|\bc\|_{\infty}^{3}
+\frac{G(\vth)}{C}x\cdot\frac{2x^{2}}{\vth_{k}^{2}}=K_{5}\,C^{2}\|\bc\|_{\infty}^{3}.
\]
Subtract, and use $K_{5}C^{2}\|\bc\|_{\infty}^{3}\le\frac\kappa2\|\bc\|_{\infty}^{2}$.
\end{proof}

\begin{corollary}[The nonperturbative ambiguity does not amplify]\label{cor:amplify}
Let $\Sm,\Sm'$ be two rules as in Proposition \ref{prop:params}, so that
$\sup_{\|\bc\|_{\infty}\le r}\|\Sm\bB-\Sm'\bB\|_{\infty}\le\Delta$ with $\Delta$ exponentially small
in $1/(Cr)$. If $\bc(\cdot),\bc'(\cdot)$ solve the two flows from the same initial datum and both
remain in $\{\|\bc\|_{\infty}\le r\}$ up to time $t$, then
\[
\|\bc(t)-\bc'(t)\|_{\infty}\ \le\ \frac{\Delta}{LCr}\Big(e^{LCr|t|}-1\Big).
\]
Over the perturbative renormalisation--group times $|t|\le\frac{1}{aCr}$ of
Corollary \ref{cor:confine} the amplification factor $e^{LCr|t|}\le e^{L/a}$ is bounded by a
constant depending only on $(k,\vth)$; the choice of summation rule therefore changes the trajectory
by an amount that is still beyond all orders in the couplings.
\end{corollary}

\begin{proof}
Grönwall's inequality applied to $u(t)=\|\bc(t)-\bc'(t)\|_{\infty}$, which satisfies
$u(0)=0$ and $\dot u\le LCr\,u+\Delta$ a.e. by Proposition \ref{prop:field}.
\end{proof}

\section{Summary of the answer}\label{sec:summary}

We have constructed, for the formal beta--function series \eqref{eq:Bdef} of four--dimensional
nonplanar Euclidean scalar $\varphi^{4}$ theory under the factorial bounds of Hypothesis \ref{hyp:H},
the summation rule
\[
\boxed{\ \Sm B(\bc)=\int_{0}^{T(\bc)}e^{-\tau}\,\hB(\tau\bc)\,\frac{d\tau}{\tau}
=\sum_{p\ge2}\Pr\big[\Poi(T(\bc))\ge p\big]\,b_{p}(\bc),\qquad
\hB(\bz)=\sum_{p\ge2}\frac{b_{p}(\bz)}{(p-1)!},\ \
T(\bc)=\frac{\vth}{C\,\nu_{k}(\bc)},\ }
\]
with the real--analytic gauge $\nu_{k}(\bc)=(\sum_{i}c_{i}^{2k})^{1/(2k)}$ and admissible parameters
$(k,\vth)$, and the domain
\[
\boxed{\ \D=\{\bc\in\R^{4}:\ C\,\nu_{k}(\bc)<\vth\}\ \supseteq\ \{\bc:\|\bc\|_{1}<\vth/C\},\ }
\]
a nonempty open star--shaped neighbourhood of the origin. On $\D$ (indeed on all of $\R^{4}$) the
rule produces a finite real number; it is linear in the coefficient family and real; the summed beta
function is real--analytic away from the origin and $C^{\infty}$ everywhere with Taylor expansion
exactly $B$; it reproduces the perturbative series to all orders with the explicit Gevrey--$1$
remainder bound of Theorem \ref{thm:master}, and its Taylor expansion at the origin is exactly $B$,
so every perturbative coefficient --- the one--loop coefficient in particular --- is recovered
exactly;
at optimal truncation its remainder is $O\big(e^{-\vth_{k}/(C\|\bc\|_{\infty})}\big)$, where the
exponent $\vth_{k}<1$ can be brought arbitrarily close to $1$ by choosing the admissible pair
suitably; and it agrees with the classical Borel sum, and with the ordinary sum of a convergent
series, up to errors that are again exponentially small in $1/(C\|\bc\|_{\infty})$ --- not exactly,
which by Theorem \ref{thm:dichotomy} is impossible for any rule defined on the whole class. Finally the construction is optimal in a precise sense: under Hypothesis
\ref{hyp:H} the series diverges at every nonzero point, its Borel transform may have a natural
boundary on $|\tau|=1/(C\|\bc\|_{\infty})$ so that no continuation--based summation exists on the
class, no exact coefficient--continuous linear rule exists, no rule defined on the whole class can be
regular, and the value of the sum is intrinsically undetermined at the level $(\pi/C)e^{-1/(C\|\bc\|_{\infty})}$. No rule can resolve the
beta function beyond that threshold, and the constructed family approaches it: every exponent
$\eta<1$ is realised, and $\eta=1$ is the exact barrier.

\appendix
\section{Elementary lemmas}\label{app:elementary}

\begin{lemma}[Elementary Stirling bounds]\label{lem:stirling}
For every integer $n\ge1$: $\ (n/e)^{n}\le n!\le e\,n^{n+1/2}e^{-n}$.
\end{lemma}

\begin{proof}
The lower bound follows from $e^{n}=\sum_{j\ge0}n^{j}/j!\ge n^{n}/n!$. For the upper bound put
$a_{n}:=n!\,e^{n}n^{-n-1/2}$; then
\[
\frac{a_{n}}{a_{n+1}}=\frac{1}{e}\Big(1+\frac1n\Big)^{n+1/2}.
\]
For $x>0$ let $\phi(x)=\ln(1+x)-\frac{2x}{2+x}$; then $\phi(0)=0$ and
$\phi'(x)=\frac{1}{1+x}-\frac{4}{(2+x)^{2}}=\frac{x^{2}}{(1+x)(2+x)^{2}}>0$, so
$\ln(1+x)>\frac{2x}{2+x}$ for $x>0$. With $x=1/n$ this gives
$(n+\frac12)\ln(1+\frac1n)>(n+\frac12)\frac{2/n}{2+1/n}=1$, i.e. $(1+\frac1n)^{n+1/2}>e$ and
$a_{n}>a_{n+1}$. Hence $a_{n}\le a_{1}=e$ for all $n\ge1$.
\end{proof}

\begin{lemma}[Least--term lemma]\label{lem:leastterm}
For $0<y\le1$,
\[
e^{-1/y}\ \le\ \min_{N\ge1}N!\,y^{N}\ \le\ e^{2}\,y^{-1/2}\,e^{-1/y}.
\]
\end{lemma}

\begin{proof}
Lower bound: by Lemma \ref{lem:stirling}, $N!y^{N}\ge(Ny/e)^{N}=e^{\phi(N)}$ with
$\phi(\nu)=\nu\ln(\nu y)-\nu$; $\phi'(\nu)=\ln(\nu y)$ vanishes only at $\nu=1/y$, where $\phi$
attains its minimum $-1/y$. Upper bound: take $N=\lfloor1/y\rfloor\ge1$, so $Ny\le1$ and
$N>1/y-1$. By Lemma \ref{lem:stirling},
$N!y^{N}\le e\sqrt N\,(Ny/e)^{N}\le e\sqrt N e^{-N}\le e\cdot y^{-1/2}\cdot e^{\,1-1/y}$.
\end{proof}

\begin{lemma}[Weierstrass convergence theorem in $\C^{n}$]\label{lem:weier}
Let $U\subseteq\C^{n}$ be open and let $f_{j}$ be holomorphic on $U$ with $f_{j}\to f$ uniformly on
compact subsets of $U$. Then $f$ is holomorphic on $U$.
\end{lemma}

\begin{proof}
$f$ is continuous. Let $\overline{\Delta}=\overline{\Delta}(a,r)\subset U$ be a closed polydisc with
distinguished boundary $\partial_{0}\Delta=\{|\zeta_{i}-a_{i}|=r_{i}\ \forall i\}$. Each $f_{j}$
satisfies the iterated Cauchy formula
$f_{j}(\bz)=(2\pi i)^{-n}\int_{\partial_{0}\Delta}f_{j}(\zeta)\prod_{i}(\zeta_{i}-z_{i})^{-1}d\zeta$
for $\bz\in\Delta$; letting $j\to\infty$, using uniform convergence on the compact set
$\partial_{0}\Delta$, the same formula holds for $f$. Expanding
$\prod_{i}(\zeta_{i}-z_{i})^{-1}$ in a multiple geometric series, uniformly convergent for $\bz$ in
a compact subset of $\Delta$ and $\zeta\in\partial_{0}\Delta$, and integrating term by term exhibits
$f$ as a convergent power series on $\Delta$.
\end{proof}

\begin{lemma}[Cauchy inequalities]\label{lem:cauchy}
Let $f$ be holomorphic on a neighbourhood of the closed polydisc $\overline{\Delta}(a,(r,\dots,r))
\subset\C^{n}$. Then for every multi--index $\alpha$,
$|\partial^{\alpha}f(a)|\le\alpha!\,r^{-|\alpha|}\sup_{\overline{\Delta}}|f|$.
\end{lemma}

\begin{proof}
The Taylor coefficients $c_{\alpha}=\partial^{\alpha}f(a)/\alpha!$ satisfy
$c_{\alpha}=(2\pi i)^{-n}\int_{\partial_{0}\Delta}f(\zeta)\prod_{i}(\zeta_{i}-a_{i})^{-\alpha_{i}-1}d\zeta$,
whence $|c_{\alpha}|\le r^{-|\alpha|}\sup_{\partial_{0}\Delta}|f|$.
\end{proof}

\begin{lemma}[Flatness lemma]\label{lem:flat}
Let $B_{r}\subseteq\R^{n}$ be an open ball centred at $0$ and let $g\in C(B_{r})\cap C^{\infty}(B_{r}\setminus\{0\})$
be such that every partial derivative $\partial^{\alpha}g$ ($\alpha\ne0$) extends continuously to
$B_{r}$; call the extension $g_{\alpha}$. Then $g\in C^{\infty}(B_{r})$ and
$\partial^{\alpha}g=g_{\alpha}$ throughout $B_{r}$.
\end{lemma}

\begin{proof}
It suffices to prove $g\in C^{1}(B_{r})$ with $\partial_{j}g=g_{e_{j}}$; the general case follows by
induction, applying that statement to $g_{e_{j}}$ (which is continuous on $B_{r}$, smooth off the
origin, and all of whose derivatives extend continuously, being $g_{\alpha+e_{j}}$). Fix
$\bc\in B_{r}\setminus\{0\}$. For $0<\eps\le1$ the segment $\{s\bc:\eps\le s\le1\}$ avoids the
origin, so
$g(\bc)-g(\eps\bc)=\int_{\eps}^{1}\sum_{j}c_{j}\,\partial_{j}g(s\bc)\,ds$. Letting $\eps\downarrow0$
and using the continuity of $g$ at $0$ and of $s\mapsto\sum_{j}c_{j}g_{e_{j}}(s\bc)$ on $[0,1]$,
\[
g(\bc)-g(0)=\int_{0}^{1}\sum_{j}c_{j}\,g_{e_{j}}(s\bc)\,ds .
\]
Hence $\big|g(\bc)-g(0)-\sum_{j}c_{j}g_{e_{j}}(0)\big|\le\|\bc\|_{1}\max_{j}\sup_{0\le s\le1}
|g_{e_{j}}(s\bc)-g_{e_{j}}(0)|=o(\|\bc\|)$ as $\bc\to0$, by continuity of the $g_{e_{j}}$ at $0$.
So $g$ is differentiable at $0$ with $\partial_{j}g(0)=g_{e_{j}}(0)$, and $\partial_{j}g=g_{e_{j}}$
is continuous on $B_{r}$.
\end{proof}

\begin{lemma}[Uniqueness of graded asymptotic expansions]\label{lem:uniqueasy}
Let $f$ be a function defined on a punctured neighbourhood of $0$ in $\R^{n}$ and let
$(b_{p})_{p\ge2}$, $(b_{p}')_{p\ge2}$ be two families of homogeneous polynomials, $\deg b_{p}=\deg b_{p}'=p$,
such that for every $N$,
$f(\bc)-\sum_{p=2}^{N}b_{p}(\bc)=o(\|\bc\|^{N})$ and $f(\bc)-\sum_{p=2}^{N}b_{p}'(\bc)=o(\|\bc\|^{N})$
as $\bc\to0$. Then $b_{p}=b_{p}'$ for all $p$. Consequently, if $f$ is real--analytic at $0$ and
admits $(b_{p})$ as its asymptotic expansion to all orders, then $\sum_{p}b_{p}$ is the Taylor series
of $f$ at $0$ and in particular converges near $0$.
\end{lemma}

\begin{proof}
Induction on $p$. If $b_{q}=b_{q}'$ for $2\le q<p$, then subtracting the two hypotheses with $N=p$
gives $b_{p}(\bc)-b_{p}'(\bc)=o(\|\bc\|^{p})$; putting $\bc=t\bu$ with $\|\bu\|=1$ and $t\downarrow0$
gives $t^{p}(b_{p}(\bu)-b_{p}'(\bu))=o(t^{p})$, so $b_{p}=b_{p}'$ on the unit sphere and hence, by
homogeneity, everywhere. For the last statement, a function real--analytic at $0$ is the sum of an
absolutely convergent power series on some polydisc; grouping its monomials by total degree produces
homogeneous polynomials $f_{p}$ with $f(\bc)-\sum_{p\le N}f_{p}(\bc)=O(\|\bc\|^{N+1})$, so
$(f_{p})$ is an asymptotic expansion of $f$ and the first part gives $f_{p}=b_{p}$.
\end{proof}

\begin{remark}
Every result of this paper has been proved from first principles, with the single exception of the
Hartogs--Osgood--Brown extension theorem quoted in Remark \ref{rem:hartogs}, which is used only in
that remark and is not needed for Theorem \ref{thm:main}. The classical existence and uniqueness
theorems for ordinary differential equations and Grönwall's inequality are used only in
\S\ref{sec:flow}.
\end{remark}

\bibliographystyle{unsrt}
\bibliography{references}

\end{document}